\documentclass{article}

\usepackage{arxiv}

\usepackage[utf8]{inputenc} 
\usepackage[T1]{fontenc}    
\usepackage{hyperref}       
\usepackage{url}            
\usepackage{xcolor}
\usepackage{booktabs}       
\usepackage{amsmath,amsthm,amssymb,bm}

\usepackage{algorithm2e}
\usepackage{comment}

\usepackage{nicefrac}       
\usepackage{microtype}      
\usepackage{lipsum}
\usepackage{graphicx}

\newcounter{theorem}
\newtheorem{thm}[theorem]{Theorem}

\newtheorem{lemma}[theorem]{Lemma}

\newtheorem{defn}[theorem]{Definition}

\newtheorem*{definition}{Definition}

\newtheorem*{remark}{Remark}

{
      \theoremstyle{plain}
      \newtheorem*{assumption1a}{(A.1)}
      \newtheorem*{assumption2a}{(A.2)}
       \newtheorem*{assumption3a}{(A.3)}
      \newtheorem*{assumption4a}{(A.4)}
      
     \newtheorem*{assumption1e}{(E.1)}
      \newtheorem*{assumption2e}{(E.2)}
       \newtheorem*{assumption3e}{(E.3)}
      \newtheorem*{assumption4e}{(E.4)}
      
       \newtheorem*{assumptionH1}{(H.1)}
      \newtheorem*{assumptionH2}{(H.2)}
             \newtheorem*{assumptionK1}{(K.1)}
      \newtheorem*{assumptionK2}{(K.2)}
  }
\newcommand{\Fa}{\mathcal{F}}

\newcommand{\E}{\mathbb{E}}

\newcommand{\var}{\mathbb{V}\mathrm{ar}} 
\newcommand{\tr}{\mbox{tr}}

\newcommand{\Var}{\mathbb{V}\mathrm{ar}}
\newcommand{\cov}{\mbox{cov}}

\newcommand{\spn}{\operatorname{span}}

\newcommand{\spc}{{\mathcal S}}

\newcommand{\real}{{\mathbb R}}
\newcommand{\1}{\mathbf 1}

\newcommand{\Z}{{\mathbf Z}}
\newcommand{\X}{{\mathbf X}}
\newcommand{\I}{{\mathbf I}}
\newcommand{\Ub}{{\mathbf U}}
\newcommand{\Ob}{{\mathbf O}}

\newcommand{\x}{{\mathbf x}}
\newcommand{\bb}{{\mathbf b}}

\newcommand{\M}{{\mathbf M}}

\newcommand{\Q}{{\mathbf Q}}

\newcommand{\0}{{\mathbf 0}}
\newcommand{\B}{{\mathbf B}}

\newcommand{\V}{{\mathbf V}}
\newcommand{\eb}{\mathbf{e}}

\newcommand{\Pbf}{\mathbf{P}}

\newcommand{\vb}{\mathbf v}
\newcommand{\bs}{\mathbf s}

\newcommand{\A}{{\mathbf A}}

\newcommand{\rs}{{\mathbf r}}
\newcommand{\y}{{\mathbf y}}

\def\xn{{\mathbf x}}

\newcommand{\Pb}{\mathbb{ P}}

\newcommand{\greekbold}[1]{\mbox{\boldmath $#1$}}

\newcommand{\iotabf}{\greekbold{\iota}}

\newcommand{\Sigmabf}{\greekbold{\Sigma}}
\newcommand{\Sigmaxbf}{\Sigmabf_{\x}}

\newcommand{\mubf}{\greekbold{\mu}}

\newcommand{\cs}{\mathcal{S}_{Y\mid\X}} 
\newcommand{\ms}{\mathcal{S}_{\E\left(Y\mid\X\right)}} 
\newcommand{\msf}{\mathcal{S}_{\E\left(f_t(Y)\mid\X\right)}} 

\newcommand{\eps}{\epsilon}
\newcommand{\ind}{\perp \!\!\! \perp}
\newcommand{\gcs}{g_{\textit{cs}}}
\newcommand{\argmin}{\operatorname{argmin}}

\title{\bf Ensemble Conditional Variance Estimator for Sufficient Dimension Reduction}

\author{Lukas Fertl\thanks{lukas.fertl@tuwien.ac.at}\\
            \hspace{.2cm}
			Institute of Statistics and Mathematical Methods in Economics\\ Faculty of Mathematics and Geoinformation\\ TU Wien, Vienna, Austria
			\And \\
			\textbf{Efstathia Bura} \thanks{efstathia.bura@tuwien.ac.at} \\
			\hspace{.2cm}
			Institute of Statistics and Mathematical Methods in Economics\\ Faculty of Mathematics and Geoinformation\\ TU Wien, Vienna, Austria}

\begin{document}
\maketitle
\begin{abstract}
\textit{Ensemble Conditional Variance Estimation} (\texttt{ECVE}) is a novel sufficient dimension reduction (SDR) method in regressions with continuous response and predictors. \texttt{ECVE} applies to general non-additive error regression models. It operates under the assumption that the predictors can be replaced by a lower dimensional projection without loss of information.It is a semiparametric forward regression model based exhaustive sufficient dimension reduction estimation method that is shown to be consistent under mild assumptions. 
It is shown to outperform  \textit{central subspace mean average variance estimation} (\texttt{csMAVE}), its main competitor, under several simulation settings  and in a benchmark data set analysis.
\end{abstract}

\section{Introduction} \label{sec:intro}
Let  $(\Omega ,{\mathcal {F}},\Pb)$ be a probability space. Let $Y$ be a univariate continuous response and $\X$ a $p$-variate continuous predictor,  jointly distributed, with $(Y,\X^T)^T:\Omega \to \real^{p+1}$. We consider the linear sufficient dimension reduction model 
\begin{align}
Y = \gcs(\B^T \X, \epsilon), \label{mod:e_basic}
\end{align}
where  $\X \in \real^p$ is independent of the random variable $\epsilon$, $\B$ is a $ p \times k$ matrix of rank $k$, 
and $\gcs: \real^{k+1} \to \real$ is an unknown non-constant function.

\cite[Thm. 1]{ZENG2010271} showed that if $(Y,\X^T)^T$ has a joint continuous distribution, \eqref{mod:e_basic} is equivalent to
\begin{align}\label{dimredspace}
Y \ind \X \mid \B^T\X,  
\end{align}
where the symbol $\ind$ indicates stochastic independence. 
The matrix $\B$ is not unique. It can be replaced by any basis of its column space, $\spn\{\B\}$. Let $\spc$ denote a subspace of $\real^p$, and let $\Pbf_{\spc}$ denote the orthogonal projection onto $\spc$ with respect to the usual inner product. If the response $Y$ and predictor vector $\X$ are independent conditionally
on $\Pbf_{\spc}\X$, then $\Pbf_{\spc}\X$ can replace $\X$ as the predictor in the regression of $Y$ on $\X$ without loss of information. Such subspaces $\spc$ are called dimension reduction subspaces and their intersection, provided it satisfies the conditional independence condition \eqref{dimredspace},
is called the central subspace and denoted by $\cs$ [see \cite[p. 105]{Cook1998}, \cite{Cook2007}]. 

By their equivalence, under both models \eqref{mod:e_basic} and \eqref{dimredspace}, $F_{Y \mid \X}(y) = F_{Y \mid \B^T\X}(y)$ and 
$\cs=\spn\{\B\}$. Since the conditional distribution of $Y \mid \X$ is the same as that of $Y \mid \B^T\X$, $\B^T\X$ contains all the information in $\X$ for modeling the target variable $Y$, and it can replace $\X$ without any loss of information.

If the error term in model \eqref{mod:e_basic} is additive with $\E(\eps \mid\X)=0$, \eqref{mod:e_basic} reduces to $Y = g(\B^T\X) + \eps$. Now, $\E(Y \mid \X) =\E(Y \mid \B^T\X)= \E(Y \mid \Pbf_{\spc} \X)$, where $\spc=\spn\{\B\}$.  The mean subspace, denoted by $\ms$, is the intersection of all subspaces $\spc$ such that $\E(Y \mid \X) = \E(Y \mid \Pbf_{\spc} \X)$ \cite{CookLi2002}.  In this case, \eqref{mod:e_basic}  becomes the classic mean subspace model with $\spn\{\B\} = \ms$. 
\cite{CookLi2002} showed that the mean subspace is a subset of the central subspace,  $ \ms \subseteq \cs$. 

Several \textit{linear sufficient dimension reduction} (SDR) methods estimate $\ms$ consistently (\cite{AdragniCook2009,MaZhu2013, Li2018,Xiaetal2002}). \textit{Linear} refers to the reduction being a linear transformation of the predictor vector. \textit{Minimum Average Variance Estimation} (\texttt{MAVE}) \cite{Xiaetal2002} is the most competitive and accurate method among them. \texttt{MAVE} differentiates from the majority of SDR methods, in that  it is not \textit{inverse regression} based such as, for example, the widely used \textit{Sliced Inverse Regression} (SIR, \cite{Li1991}).  \texttt{MAVE} requires minimal assumptions on the distribution of $(Y,\X^T)^T$ and  is based on estimating the gradients of the regression function $E(Y \mid \X)$ via local-linear smoothing \cite{locallinearsmoothing}.

The \textit{central subspace mean average variance estimation} (\texttt{csMAVE}) \cite{WangXia2008,MAVEpackage} is the extension of \texttt{MAVE}  that consistently and exhaustively estimates the $\spn\{\B\}$ in model \eqref{mod:e_basic} without restrictive assumptions limiting its applicability. \texttt{csMAVE} has remained the gold standard since it was proposed by \cite{WangXia2008}. It is based on repeatedly applying  \texttt{MAVE} 
on the sliced target variables $f_u(Y) = \1_{\{s_{u-1} < Y \leq s_u\}}$ for $s_1 < \ldots < s_H$. \cite{WangXia2008} showed that the mean subspaces of the sliced $Y$ can be combined to recover the central subspace $\cs$.

Several papers made  contributions in establishing a road path from the central mean to the central subspace [see \cite{YinLi2011} for a list of references]. \cite{YinLi2011} recognized that these approaches pointed to the same direction: if one can estimate
the central mean subspace of $\E(f(\X) \mid Y)$ for sufficiently many functions $f \in \Fa$ for a family of functions $\Fa$,
then one can recover the central subspace. Such families that are rich enough to obtain the desired outcome are called \textit{characterizing ensembles} by \cite{YinLi2011}, who also proposed and  studied such functional families [see also \cite{Li2018} for an overview]. 

In this paper, we extend the \textit{conditional variance estimator} (\texttt{CVE}) \cite{FertlBura} to the exhaustive \textit{ensemble conditional variance estimator} for recovering fully the \textit{central subspace} $\cs$. \textit{Conditional variance estimation} is a semi-parametric method for the estimation of $\ms$ consistently under minimal regularity assumptions on the distribution of $(Y,\X^T)^T$. In contrast to other SDR approaches, it operates by identifying the orthogonal complement of $\ms$.
In this paper we apply the \textit{conditional variance estimator} (\texttt{CVE}) to identify the mean subspace $\msf$ of transformed responses $f_t(Y)$, where $f_t$ are elements of an \textit{ensemble} $\Fa = \{f_t : t \in \Omega_T\}$, and then combine them to form the \textit{central subspace} $\cs$.

The paper is organized as follows. In Section~\ref{Preliminaries} we define the notation and concepts we use throughout the paper. A short overview of ensembles is given in Section~\ref{e_motivation}. The \textit{ensemble conditional variance estimator} (ECVE) is introduced in Section~\ref{sec:ensembleCVE} and the estimation procedure in Section~\ref{e_estimation}. In  Section~\ref{sec:consistency}, the consistency of the \textit{ensemble conditional variance estimator} for the central subspace is shown. We  assess and compare the performance of the estimator vis-a-vis \texttt{csMAVE} via simulations in Section~\ref{sec:simulations} and by applying it to the Boston Housing data in Section~\ref{sec:dataAnalysis}. We conclude in Section~\ref{sec:discussion}.

\section{Preliminaries}\label{Preliminaries}
We denote by $F_\Z$  the cumulative distribution function (cdf) of a random variable or vector $\Z$. We drop the subscript, when the attribution is clear from the context. For a matrix $\A$,  $\|\A\|$ denotes its Frobenius norm, and $\|\bm{a}\|$ the Euclidean norm for a vector $\bm{a}$.  Scalar product refers to the usual Euclidean scalar product, and $\perp$ denotes orthogonality with respect to it. The probability density function of $\X$ is denoted by $f_{\X}$, and its support by $\mbox{supp}(f_{\X})$. The notation $Y \ind \X$ signifies  stochastic independence of the random vector $\X$ and random variable $Y$. The $j$-th standard basis vector with zeroes everywhere except for 1 on the $j$-th position is denoted by  $\eb_j \in \real^p$, $\iotabf_p = (1,1,\ldots,1)^T \in \real^p$, and $\I_p = (\eb_1,\ldots,\eb_p)$ is the  identity matrix of order $p$. For any matrix $\M \in \real^{p \times q}$, $\Pbf_{\M}$ denotes the orthogonal projection matrix on its column or range space $\spn\{\M\}$;  i.e., $\Pbf_{\M} = \Pbf_{\spn\{\M\}} =\M(\M^T \M)^{-1} \M^T \in \real^{p \times p}$.

For $q \leq p$,  
\begin{equation}\label{Smanifold}
    \spc(p,q) = \{\V \in \real^{p \times q}: \V^T\V = \I_q\},
\end{equation}
denotes the Stiefel manifold that comprizes of all $p \times q$ matrices with orthonormal columns.  $\spc(p,q)$ is compact with  $\dim(\spc(p,q)) = pq - q(q+1)/2$  [see  \cite{Boothby} and Section 2.1 of \cite{Tagare2011}]. 
The set
\begin{equation}\label{Grassman_def}
    Gr(p,q) = \spc(p,q)/\spc(q,q)
\end{equation}
denotes a Grassmann manifold \cite{Grassman} that contains all $q$-dimensional subspaces in $\real^p$. $Gr(p,q)$ is  the quotient space of $\spc(p,q)$ with all $q \times q$ orthonormal matrices in $\spc(q,q)$. 

\subsection{Ensembles}\label{e_motivation}

\cite{YinLi2011} introduced \textit{ensembles} as a device to extend mean subspace  to central subspace SDR methods.
The \textit{ensemble} approach of combining mean subspaces to span the central subspace comprizes of two components: (a) a rich family of functions of transformations for the response and (b) a sampling mechanism for drawing the functions from the ensemble to ascertain coverage of the central subspace. To distinguish between families of functions and ensembles, \cite{YinLi2011} use the term \textit{parametric} ensemble, which we define next. 

\begin{definition}\label{parametricensemble}
A family $\Fa$ of measurable functions from $\real$ to $\real$ is called an ensemble. If $\Fa$ is a  family of measurable functions with respect to an index set $\Omega_T$; i.e. $\Fa = \{f_t : t \in \Omega_T\}$, $\Fa$ is called a parametric ensemble.
\end{definition}

\medskip
Let $\Fa$ be an ensemble, $f \in \Fa$ and let $f(Y)$, for $Y$ following  model~\eqref{mod:e_basic}. The space $\spc_{\E(f(Y)\mid \X)}$ is defined to be the mean subspace of the transformed random variable $f(Y)$ [see \cite{Cook1998} or \cite{CookLi2002}]. 

\begin{definition}
An ensemble $\Fa$ characterizes the central subspace $\cs$, if 
\begin{align}\label{Fa_characterises_cs}
    \spn\{\msf: f_t \in \Fa\} = \cs
\end{align}
\end{definition}

As an example, the parametric ensemble $\Fa =\{f_t: t \in \Omega_T\} = \{1_{\{z \leq t\}}: t \in \real\}$ can characterize the \textit{central subspace} $\cs$. That is, $\E(f_t(Y)|\X)$ is the conditional cumulative distribution function evaluated at $t$. 
To see this, let $\B \in \spc(p,k)$ be such that $\E(f_t(Y) \mid \X) = \E(f_t(Y) \mid \B^T \X)$ for all $t$.  Then,
$F_{Y \mid \X}(t) = \E(f_t(Y) \mid \X) = \E(f_t(Y) \mid \B^T \X) = F_{Y \mid \B^T\X}(t)$ for all $t$.  
Varying over the parametric ensemble $\Fa$, in this case over $t \in \real$, obtains the conditional cumulative distribution function.  
This \texttt{indicator} ensemble  fully recovers the conditional distribution of $Y\mid \X$ and, thus, also the \textit{central subspace} $\cs$,
\begin{align*}
  \spn\{\msf: f_t \in \Fa\} =  \spn\{\mathcal{S}_{\E\left(1_{\{Y \leq t\}} \mid \X \right)}: t \in \real\} = \cs
\end{align*}
We reproduce a list of parametric ensembles $\Fa$, and associated regularity conditions, that can characterize $\cs$ from \cite{YinLi2011} next. 
\begin{description}
    \item[Characteristic ensemble] $\Fa =\{f_t: t \in \Omega_T\} = \{\exp(i t \cdot): t \in \real\}$  
    \item[Indicator ensemble] $\Fa = \{1_{\{z \leq t\}}: t \in \real\}$, where $\spn\{\msf: f_t \in \Fa\}$ recovers the conditional cumulative distribution function
    \item[Kernel ensemble] $\Fa = \{ h^{-1}K\left((z-t)/h\right): t \in \real, h > 0\}$,  where $K$ is a kernel suitable for density estimation, and  $\spn\{\msf: f_t \in \Fa\}$ recovers the conditional density
    \item[Polynomial ensemble] $\Fa = \{z^t: t = 1,2,3,...\}$, where $\spn\{\msf: f_t \in \Fa\}$ recovers the conditional moment generating function
    \item[Box-Cox ensemble] $\Fa = \{(z^t-1)/t : t \neq 0\} \cup \{ \log(z): t = 0\}$ 
    Box-Cox Transforms
    \item[Wavelet ensemble] Haar Wavelets 
\end{description}

The \texttt{characteristic} and  \texttt{indicator} ensembles describe the conditional characteristic and distribution function of $Y\mid \X$, respectively, which always exist and determine the distribution uniquely. If the conditional density function $f_{Y\mid \X}$ of $Y\mid \X$ exists, then the \textit{kernel} ensemble  characterizes the conditional distribution $Y\mid \X$. Further, if the conditional moment generating function 
exists, then the polynomial ensemble characterizes $\cs$.
\cite{YinLi2011} used the ensemble device to extend \texttt{MAVE} \cite{Xiaetal2002}, which targets the mean subspace, to its ensemble version that also estimates the central subspace $\cs$ consistently. 

\medskip
Theorem~\ref{Fa_characteris_cs_thm} \cite[Thm~2.1]{YinLi2011} establishes when an ensemble $\Fa$ is rich enough to characterize $\cs$.

\begin{thm}\label{Fa_characteris_cs_thm}
Let $\mathcal{B} = \{1_A: A \text{ is a Borel set in}\; \text{supp}(Y) \}$ be the set of indicator functions on $\text{supp}(Y)$ and $ L^2(F_Y)$ be the set of square integrable random variables with respect to the distribution $F_Y$ of the response $Y$. If $\Fa \subseteq L^2(F_Y)$ is dense in $\mathcal{B} \subseteq L^2(F_Y)$, then the ensemble $\Fa$ characterizes the \textit{central subspace} $\cs$.
\end{thm}

In Theorem~\ref{finite_characterisation_of_cs} we show that finitely many functions of an ensemble $\Fa$ are sufficient to characterize the \textit{central subspace} $\cs$.  

\begin{thm}\label{finite_characterisation_of_cs}
If a parametric ensemble $\Fa$ characterizes $\cs$, then there exist finitely many functions $f_t \in \Fa$, with $t = 1,\ldots,m$ and $m \in \mathbb{N}$, such that 
\begin{align*}
\spn\{\msf: t \in 1,\ldots,m\} = \cs
\end{align*}
\end{thm}
\begin{proof}[Proof:]
Let $k = \dim(\cs) \leq p$. Since  $\Fa$ characterizes $\cs$,  $\dim(\msf) = k_t \leq k$ by  \eqref{Fa_characterises_cs} for any $t$. If $k_t = 0$, $\msf = \{\0\}$ so the corresponding $f_t$  does not contribute to \eqref{Fa_characterises_cs}. Assume $k_t \ge 1$.  
If there were infinitely many $\msf \neq \{\0\}$ of dimension at least 1, whose span is $\cs$, then infinitely as many are identical, otherwise the dimension of the central subspace $\cs$ would be infinite,   contradicting that $\dim(\cs)=k <\infty$. 
\end{proof}

The importance of Theorem \ref{finite_characterisation_of_cs} lies in the fact that the search to characterize the central subspace is over a finite  set, even though  it does not offer tools for identifying the elements of the ensemble.  

\section{Ensemble CVE}\label{sec:ensembleCVE}

Throughout the paper, we refer to the following assumptions as needed.

\begin{assumption1e}
Model \eqref{mod:e_basic}, $Y = \gcs(\B^T\X, \epsilon)$ 
holds with $Y \in \real$, $\gcs:\real^k \times \real \to \real$  non constant in the first argument, $\B= (\bb_1, ..., \bb_k) \in \spc(p,k)$, $\X \in \real^p$ is independent of $\epsilon$, the distribution of $\X$ is absolutely continuous with respect to the Lebesgue measure in $\real^p$, $\text{supp}(f_\X)$ is convex, and $ \var(\X) = \Sigmaxbf $ is positive definite.

\end{assumption1e}

\begin{assumption2e}
The density $f_\X : \real^p \to [0,\infty)$ of $\X$ is twice continuously differentiable with compact support $\text{supp}(f_\X)$.
\end{assumption2e}
\begin{assumption3e}
For a parametric ensemble $\Fa$, its index set $\Omega_T$ is endowed with a probability measure $F_T$ such that for all $t \in \Omega_T$ with $ \msf \neq \{\0\}$,
\begin{align*}
    \Pb_{F_T} \left( \{\tilde{t} \in \Omega_T: \spc_{\E(f_{\tilde{t}}(Y)\mid \X)} = \msf\} \right) > 0
\end{align*} 
\end{assumption3e}

\begin{assumption4e}
For an ensemble $\Fa$ we assume that for all $f \in \Fa$, the conditional expectation 
\begin{align*}
    \E\left(f(Y) \mid \X \right)
\end{align*}
is twice continuously differentiable in the conditioning 
argument. Further, for all $f \in \Fa$
\begin{align*}
    \E(|f(Y)|^8) < \infty
\end{align*}
\end{assumption4e}

Assumption (E.1) 
    assures the existence and uniqueness of  $\cs = \spn\{\B\}$. Furthermore, it allows the mean subspace to be a proper subset of the central subspace, i.e. $\ms \subsetneq \cs$.
In Assumption (E.2), 
    the compactness assumption for $\text{supp}(f_\X)$ is not as restrictive as it might seem. \cite[Prop. 11]{CompactAssumption} showed that there is a compact set $K \subset \real^p$ such that $\mathcal{S}_{Y \mid \X_{|K}} = \cs$, where $\X_{|K} = \X 1_{\{\X \in K\}}$.
Assumption (E.3) simply states that the set of indices that characterize the \textit{central subspace} $\cs$ is not a null set. In practice, the choice of the probability measure $F_T$ on the index set $\Omega_T$ of a parametric ensemble $\Fa$ can always guarantee the fulfillment of  this assumption.
If the characteristic or indicator ensemble are used, (E.4) states that the conditional characteristic or distribution function are twice continuously differentiable. In this case, the 8$th$ moments  exist since  the complex exponential and indicator functions are bounded.
\medskip
\begin{definition}
For $q \leq p \in \mathbb{N}$, $f \in \Fa$, and any $\V \in S(p,q)$, we define
\begin{equation}
\tilde{L}_\Fa(\V, \bs_0,f) = \Var\left(f(Y)\mid \X \in \bs_0 + \spn\{\V\}\right) \label{e_Lvs}
\end{equation}
where $\bs_0 \in \real^p$ is a non-random shifting point.
\end{definition}

\begin{definition}
Let $\Fa$ be a parametric ensemble and $F_T$ a cumulative distribution function (cdf) on the index set $\Omega_T$. For $q \leq p$, and any $\V \in S(p,q)$, we define
\begin{align}\label{e_objective}
L_\Fa(\V) &= \int_{\Omega_T}\int_{\real^p}\tilde{L}(\V,\x, f_t)d F_{\X}(\x) d F_T(t)  \\
&= \E_{t \sim F_T} \left(\mathbb{E}_\X\left(\tilde{L}_\Fa(\V,\X,f_t)\right)\right) = \E_{t \sim F_T}(L_\Fa^*(\V,f_t)), \notag
\end{align} 
where $F_{\X}$ is the cdf of $\X$,
and 
\begin{align}\label{e_LV1}
L_\Fa^*(\V,f_t) = \mathbb{E}_\X\left(\tilde{L}_\Fa(\V,\X,f_t)\right).
\end{align}
\end{definition}

For the identity function, $f_{t_0}(z) = z$, \eqref{e_LV1} is the target function of the  \textit{conditional variance estimation} proposed in \cite{FertlBura}. 
If the random variable $t$ is concentrated on $t_0$; i.e., $t \sim \delta_{t_0}$, then the \textit{ensemble conditional variance estimator} (\texttt{ECVE}) coincides with the  \textit{conditional variance estimator} (CVE). 

The following theorem will be used in establishing the main result of this paper, which obtains the \textit{exhaustive} sufficient reduction of the conditional distribution  of $Y$ given the predictor vector $\X$.

\begin{thm}\label{Y_decomposition_thm}
Assume (E.1) and (E.2) hold, in particular model \eqref{mod:e_basic} holds. Let $\widetilde{\B}$ be a basis of $\msf$; i.e. $\spn\{\widetilde{\B}\} = \msf \subseteq \cs = \spn\{\B\}$. Then, for any $f \in \Fa$ for which assumption (E.4) holds, 
\begin{align}
    f(Y) = g(\widetilde{\B}^T\X) + \tilde{\eps}, \label{Y_decomposition}
\end{align}
with $\E(\tilde{\eps}\mid \X) = 0$ and $g:\real^{k_t} \to \real$ is a twice continuously differentiable function, where $k_t = \dim(\msf)$. 
\end{thm}

By Theorem~\ref{Y_decomposition_thm}, any  response $Y$ can be written as an additive error via the decomposition \eqref{Y_decomposition}. The predictors and the additive error term are only required to be conditionally uncorrelated in model \eqref{Y_decomposition}. The \textit{conditional variance estimator} \cite{FertlBura} also estimated $\widetilde{\B}$ in \eqref{Y_decomposition} but under the more restrictive condition of predictor and error independence.  

\begin{proof}[Proof of Theorem~\ref{Y_decomposition_thm}]

\begin{align*}
   f(Y)  &= \E\left( f(Y) \mid \X \right) + \underbrace{f(Y) - \E\left(f(Y) \mid \X \right)}_{\tilde{\epsilon}} = \E\left(f(Y) \mid \X \right) + \tilde{\eps} \notag \\
    &= \E\left(f(Y) | \widetilde{\B}^T\X \right) + \tilde{\eps} = g(\widetilde{\B}^T\X) + \tilde{\eps}
\end{align*}
where $g(\widetilde{\B}^T\X)=\E\left(f(Y) | \widetilde{\B}^T\X \right)$. 
By the tower property of the conditional expectation, $\E(\tilde{\eps}\mid \X) = \E(f(Y)\mid \X) - \E(\E(f(Y)\mid \X)\mid \X) = \E(f(Y)\mid \X) - \E(f(Y)\mid \X) = \0$. The function  $g$ is twice continuous differentiable by (E.4).
\end{proof}

\begin{thm}\label{CVE_targets_meansubspace_thm}
Assume (E.1) and (E.2) hold. Let $\Fa$ be a parametric ensemble, $\bs_0 \in \text{supp}(f_\X) \subset \real^p$, $\V \in S(p,q)$ defined in \eqref{Smanifold}. Then, for any $f \in \Fa$ for which assumption (E.4) holds, 
\begin{align}\label{e_LtildeVs0}
\tilde{L}_\Fa(\V,\bs_0,f) = \mu_2(\V,\bs_0,f) - \mu_1^2(\V,\bs_0,f) + \var(\tilde{\eps}\mid \X \in \bs_0 + \spn\{\V\})
\end{align}
where 
\begin{equation}\label{mu_l}
\mu_l(\V,\bs_0,f) = \int_{\real^q}  g(\widetilde{\B}^T\bs_0 + \widetilde{\B}^T\V\rs_1)^l\frac{f_\X(\bs_0 + \V\rs_1)}{\int_{\real^q}f_\X(\bs_0 + \V\rs)d\rs} d\rs_1 = \frac{t^{(l)}(\V,\bs_0,f)}{t^{(0)}(\V,\bs_0,f)},
\end{equation}
for $g$ given in \eqref{Y_decomposition} with 
\begin{equation}\label{tl}
t^{(l)}(\V,\bs_0,f) = \int_{\real^q}  g(\widetilde{\B}^T\bs_0 + \widetilde{\B}^T\V\rs_1)^l f_\X(\bs_0 + \V\rs_1) d\rs_1,
\end{equation}
and 
\begin{gather}
    \var(\tilde{\eps}\mid \X \in \bs_0 + \spn\{\V\}) = \E(\tilde{\eps}^2\mid \X \in \bs_0 + \spn\{\V\}) \notag\\
    = \int_{\text{supp}(f_\X)\cap\real^q} h(\bs_0 + \V \rs_1 ) f_\X(\bs_0 + \V \rs_1)d\rs_1 / \int_{\real^q}f_\X(\bs_0 + \V\rs)d\rs    = \frac{\tilde{h}(\V,\bs_0,f)}{t^{(0)}(\V,\bs_0,f)} \label{tilde_eps_var} 
\end{gather}
with $\E(\tilde{\eps}^2\mid \X = \xn) = h(\xn)$ and $\tilde{h}(\V,\bs_0,f) = \int_{\text{supp}(f_\X)\cap\real^q} h(\bs_0 + \V \rs_1 ) f_\X(\bs_0 + \V \rs_1)d\rs_1$. Further assume $h(\cdot)$ to be continuous, then $L_\Fa^*(\V,f_t)$ in \eqref{e_LV1} is well defined and continuous, 
\begin{align}\label{CVE_of_transformed_Y}
\V^t_q = \argmin_{\V \in \spc(p,q)}L_\Fa^*(\V,f_t)
\end{align}
is well defined, and the conditional variance estimator of the transformed response $f_t(Y)$ identifies $\msf$,
\begin{align}\label{CVE_targets_meansubspace}
 \msf = \spn\{\V^t_q\}^\perp.
\end{align}
\end{thm}

\cite{FertlBura} assumed model $Y = g(\B^T\X) + \eps$ with $\eps \ind \X$, which implies $\ms = \spn\{\B\} = \cs$.  \cite{FertlBura} showed that the \textit{conditional variance estimator (CVE)} can identify $\ms$ at the population level.

Theorem~\ref{CVE_targets_meansubspace_thm} extends this result to obtain that the \textit{conditional variance estimator (CVE)} identifies the \textit{mean subspace} $\ms$ also in models of the form $Y = g(\B^T\X) + \tilde{\eps}$, where $\tilde{\eps}$ is simply conditionally uncorrelated 
with $\X$. This allows CVE to apply to problems where the \textit{mean subspace} is a proper subset of the \textit{central subspace}, i.e. $\ms \subsetneq \cs$.

$\V^t_q$ in \eqref{CVE_of_transformed_Y} is not unique since for all orthogonal $\Ob \in \real^{q \times q}$, 
$L_\Fa^*(\V^t_q \Ob,f_t) = L_\Fa^*(\V^t_q ,f_t)$ as $L_\Fa^*(\V^t_q ,f_t)$ depends on $\V^t_q$ only through $\spn\{\V^t_q\}$ by \eqref{e_Lvs}. Nevertheless, it is a unique minimizer over the Grassmann manifold $Gr(p,q)$ in \eqref{Grassman_def}. To see this, suppose $\V \in \spc(p,q)$ is an arbitrary basis of a subspace $\M \in Gr(p,q)$. We can identify $\M$ through the projection $\Pbf_\M = \V\V^T$. 
By \eqref{ortho_decomp}, we write $\xn = \V \rs_1 + \Ub \rs_2$. Application of the Fubini-Tonelli Theorem yields 
\begin{align}\label{Grassman}
    \tilde{t}^{(l)}(\Pbf_\M,\bs_0,f) &= \int_{\text{supp}(f_\X)} g(\B^T\bs_0 + \B^T \Pbf_\M \xn)^l f_\X(\bs_0 + \Pbf_\M \xn)d\xn \\&= t^{(l)}(\V,\bs_0,f) \int_{\text{supp}(f_\X)\cap \real^{p-q} }d\rs_2. \notag
\end{align}
Therefore $\tilde{t}^{(l)}(\Pbf_\M,\bs_0,f)/\tilde{t}^{(0)}(\Pbf_\M,\bs_0,f) = t^{(l)}(\V,\bs_0,f)/t^{(0)}(\V,\bs_0,f)$ and $\mu_l(\cdot,\bs_0,f)$ in \eqref{mu_l} can also be viewed as a function from $Gr(p,q)$ to $\real$.

\medskip
Next we define the \textit{ensemble conditional variance estimator (ECVE)} for a parametric ensemble $\Fa$ which characterizes the \textit{central subspace} $\cs$. Following the \textit{ensemble minimum average variance estimation} formulation in \cite{YinLi2011}, we extend the original objective function by integrating over the index random variable $t \sim F_T$ in \eqref{e_objective} that indexes the ensemble $\Fa$ as \cite{YinLi2011}. 

\begin{defn}
Let 
\begin{align}\label{enVq}
\V_q = \argmin_{\V \in S(p,q)}L_\Fa(\V)
\end{align}
The \textbf{Ensemble Conditional Variance Estimator}  with respect to the ensemble $\Fa$  is defined to be any basis $\B_{p-q,\Fa}$ of $\spn\{\V_q\}^\perp$.
\end{defn}

\begin{thm}\label{ECVE_identifies_cs_thm}
Assume (E.1), (E.2), (E.3), and (E.4) hold, and that the function $h(\cdot)$ defined in Theorem~\ref{CVE_targets_meansubspace_thm} is continuous. Let $\Fa$ be a parametric ensemble that characterizes $\cs$, with $k = \dim(\cs)$, and $\V$ be an element of the Stiefel manifold $S(p,q)$, which is defined in \eqref{Smanifold}, with $q = p - k$. Then, $\V_q$ in \eqref{enVq}
is well defined and
\begin{align}
\cs = \spn\{\V_q\}^\perp.
\end{align}
\end{thm}

\section{Estimation of the ensemble CVE}\label{e_estimation}
Assume $(Y_i,\X_i^\top)_{i=1,...,n}^\top$ is an i.i.d. sample  from model \eqref{mod:e_basic}, and let
\begin{align}
	d_i(\V,\bs_0)&= \|\X_i - \Pbf_{\bs_0 + \spn\{\V\}}\X_i\|_2^2 = \|\X_i -\bs_0\|_2^2 - \langle \X_i - \bs_0,\V\V^\top(\X_i - \bs_0)\rangle  \notag\\
	&= \| (\I_p - \V\V^\top)(\X_i - \bs_0)\|_2^2 = \| \Q_{\V}(\X_i - \bs_0)\|_2^2 \label{distance}
\end{align}
where $\langle \cdot, \cdot\rangle$ is the usual inner product in $\real^p$, $\Pbf_{\V}=\V\V^\top$ and $\Q_{\V}=\I_p-\Pbf_{\V}$. The estimators we propose involve a variation of kernel smoothing, which depends on a bandwidth $h_n$. In our procedure, $h_n$ is the squared width of a slice around the subspace $\bs_0 + \spn\{\V\}$.  In order to obtain pointwise convergence for the ensemble CVE, we use the following bias and variance assumptions on the bandwidth, as typical in nonparametric estimation.

\begin{assumptionH1}
For $n \to \infty$, $h_n \to 0$
\end{assumptionH1}
\begin{assumptionH2}
For $n \to \infty$, $nh^{(p-q)/2}_n \to \infty$
\end{assumptionH2}

In order to obtain consistency of the proposed estimator, Assumption (H.2) will be strengthened to $\log(n)/nh^{(p-q)/2}_n \to 0$.

We also let $K$, which we refer to as \textit{kernel}, be a  function satisfying the following assumptions.
\begin{assumptionK1}
$K:[0,\infty) \rightarrow [0,\infty)$ 
is a non increasing and continuous 
function, so that $|K(z)| \leq M_1$, with $\int_{\real^{q}} K(\|\rs\|^2) d\rs < \infty$ for $q \leq p-1$.
\end{assumptionK1}
\begin{assumptionK2}
There exist positive finite constants $L_1$ and $L_2$ such that $K$ satisfies either (1) or (2) below:
\begin{itemize}
    \item[(1)] $K(u) = 0$ for $|u| > L_2$ and for all $u, \tilde{u}$ it holds $|K(u) - K(\tilde{u})| \leq L_1 |u - \tilde{u}|$
    \item[(2)] $K(u)$ is differentiable with $|\partial_u K(u)| \leq L_1$ and for some $\nu > 1$ it holds $|\partial_u K(u)| \leq L_1 |u|^{-\nu}$ for $|u| > L_2$
\end{itemize}
\end{assumptionK2}
The Gaussian kernel $K(z) = \exp(-z^2)$, for example, fulfills both (K.1) and (K.2) [see \cite{Hansen2008}], and will be used throughout the paper.

For $i=1,\ldots,n$, we let
\begin{equation}
w_i(\V,\bs_0) = \frac{K\left(\frac{d_i(\V,\bs_0)}{h_n}\right)}{\sum_{j=1}^nK\left(\frac{d_j(\V,\bs_0)}{h_n}\right)} \label{weights}
\end{equation}
\begin{equation} \label{e_ybar}
\bar{y}_l(\V,\bs_0,f) = \sum_{i=1}^n w_i(\V,\bs_0)f(Y_i)^l \quad \text{for} \quad l=1,2
\end{equation}
We estimate  $\tilde{L}_\Fa(\V,s_0,f)$ in \eqref{e_LtildeVs0} with
\begin{equation}
\tilde{L}_{n,\Fa}(\V,s_0, f)  = \bar{y}_2(\V,\bs_0,f) - \bar{y}_1(\V,\bs_0,f)^2, \label{e_Ltilde}
\end{equation}
and the objective function $L^*_\Fa(\V,f)$ in \eqref{e_LV1} with 
\begin{equation}
L^*_n(\V,f) = \frac{1}{n} \sum_{i=1}^n \tilde{L}_{n,\Fa}(\V,\X_i,f), \label{e_LN}
\end{equation}
where each data point $\X_i$ is a shifting point. For a parametric ensemble  $\Fa = \{f_t : t \in \Omega_T\}$ and $(t_j)_{j=1,...,{m_n}}$ an i.i.d. sample from $F_T$ with $\lim_{n \to \infty} m_n = \infty$, the final estimate of the objective function in \eqref{e_objective} is given by 
\begin{equation}\label{e_objective_est}
L_{n,\Fa}(\V) = \frac{1}{m_n} \sum_{j=1}^{{m_n}} L^*_n(\V,f_{t_j})
\end{equation}
The ensemble conditional variance estimator (ECVE) is defined to be any basis of $\spn\{\hat{\V}_q\}^\perp$, where
\begin{equation} \label{e_optim}
    \hat{\V}_q = \argmin_{\V \in S(p,q)}L_{n,\Fa}(\V)
\end{equation}
We use the same algorithm as in \cite{FertlBura} to solve the optimization problem \eqref{e_optim}. It requires the explicit form of the gradient of \eqref{e_objective_est}. Theorem~\ref{e_lemma-one} provides the gradient when a Gaussian kernel is used.

\begin{thm}\label{e_lemma-one}
The gradient of $\tilde{L}_{n,\Fa}(\V,s_0, f) $ in \eqref{e_Ltilde} is given by 
\begin{align*}
\nabla_{\V}\tilde{L}_{n,\Fa}(\V,s_0, f) = \frac{1}{h_n^2}\sum_{i=1}^n (\tilde{L}_{n,\Fa}(\V,\bs_0, f) - (f(Y_i)-\bar{y}_1(\V,\bs_0,f))^2)w_id_i\nabla_{\V}d_i(\V,\bs_0) \in \real^{p \times q},
\end{align*}
and the gradient of $L_{n,\Fa}(\V)$ in \eqref{e_objective_est} is 
\[
\nabla_{\V}L_{n,\Fa}(\V) = \frac{1}{n {m_n}} \sum_{i=1}^n \sum_{j=1}^{{m_n}} \nabla_{\V}\tilde{L}_{n,\Fa}(\V,\X_i,f_{t_j}).
\]
\end{thm}

In the implementation of ECVE, we follow \cite{FertlBura} and set the bandwidth to
\begin{equation}
  \label{bandwidth}
h_n =  1.2^2 \frac{2\tr(\widehat{\Sigma}_\xn)}{p} \left(n^{-1/(4+p-q)} \right)^2.
\end{equation}
where $\widehat{\Sigma}_\xn = (1/n) \sum_i (\X_i -\Bar{\X})(\X_i -\Bar{\X})^T$ and $\Bar{\X} = (1/n) \sum_i \X_i$. 

\subsection{Weighted estimation of $L^*_n(\V,f)$}\label{weight_section}

The set of points $\{\xn \in \real^p: \|\xn - \Pbf_{\bs_0 + \spn\{\V\}}\xn\|^2 \leq h_n\}$ represents a \textit{slice} in the subspace of $\real^p$ about $\bs_0+ \spn\{\V\}$. 
In the estimation of $L(\V)$ two different weighting schemes are used:
    (a) \textit{Within slices}: The weights are defined in \eqref{weights} and are used to calculate \eqref{e_Ltilde}.
    (b) \textit{Between slices}: Equal weights $1/n$ are used to calculate \eqref{e_LN}.
Another idea for the between slices weighting is to assign more weight to slices with more points. This can be realized by altering \eqref{e_LN} to
\begin{align}
L^{(w)}_n(\V,f) &=  \sum_{i=1}^n \tilde{w}(\V,\X_i) \tilde{L}_n(\V,\X_i,f), \quad \mbox{with} \label{wLN}\\
\tilde{w}(\V,\X_i) &= \frac{\sum_{j=1}^n K(d_j(\V,\X_i)/h_n) - 1}{\sum_{l,u=1}^nK(d_l(\V,\X_u)/h_n) -n} = \frac{\sum_{j=1,j\neq i}^n K(d_j(\V,\X_i)/h_n) }{\sum_{l,u=1, l\neq u}^nK(d_l(\V,\X_u)/h_n)}\label{wtilde}
\end{align}
The denominator in \eqref{wtilde} 
guarantees the weights $\tilde{w}(\V,\X_i)$ sum up to one.
If \eqref{wLN} instead of \eqref{e_LN} is used in \eqref{e_objective_est} we refer to this method as \textit{weighted ensemble conditional variance estimation}.

For example, if a rectangular kernel is used, $\sum_{j=1,j\neq i}^n K(d_j(\V,\X_i)/h_n)$ is the number of $\X_j$ ($j \neq i$) points in the slice corresponding to $\tilde{L}_n(\V,\X_i,f)$. Therefore, this slice is assigned  weight that is proportional to the number of $\X_j$ points in it, and the more observations we use for estimating $L(\V,\X_i,f)$, the better its accuracy.

\section{Consistency of the ECVE}\label{sec:consistency}

The consistency of ECVE derives from the consistency of CVE \cite{FertlBura} that targets a specific $\msf$ and the fact that we can recover $\cs$ from $\msf$ across all transformations $f_t \in \Fa = \{f_t : t \in \Omega_T\}$ for an ensemble that characterizes $\cs$. 
This is achieved in sequential steps from  Theorem \ref{uniform_convergence_ecve}, which is the main building block, to Theorem \ref{ECVE_consistency}. The proofs are technical and lengthy, and, thus, are given in the Appendix.

\begin{thm}\label{uniform_convergence_ecve}
Assume conditions (E.1), (E.2), (E.4), (K.1), (K.2), (H.1) hold, $a_n^2 = \log(n)/nh_n^{(p-q)/2} = o(1)$, and $a_n/h_n^{(p-q)/2} = O(1)$. Let $\Fa$ be a parametric ensemble such that  $\E(|\tilde{\eps}|^l \mid\X =\xn)$ is continuous for $l = 1,\ldots,4$, and the second conditional moment is twice continuously differentiable, where $\tilde{\eps}$ is given by Theorem~\ref{Y_decomposition_thm}. Then, $L^*_n(\V,f)$, defined in \eqref{e_LN}, converges uniformly in probability to $L^*(\V,f)$  in \eqref{e_LV1} for all $f \in \Fa$; i.e.,
\begin{align*}
    \sup_{\V \in \spc(p,q)}|L^*_n(\V,f) -L^*(\V,f)| \longrightarrow 0 \quad \text{in probability as $n \to \infty$.}
\end{align*}
\end{thm}
Next, Theorem~\ref{thm_consistency_mean_subspace} shows that ensemble conditional variance estimator is consistent for $\msf$ for any transformation $f$. 

\begin{thm}\label{thm_consistency_mean_subspace}
Under the same conditions as Theorem~\ref{uniform_convergence_ecve}, 
the conditional variance estimator $\spn\{\widehat{\B}^t_{{k_t}}\}$ estimates $\msf$ consistently, for $f_t \in \Fa$. That is,
\begin{equation*}
\|\Pbf_{\widehat{\B}^t_{{k_t}}} - \Pbf_{\msf}\| \to 0 \quad \text{in probability as } n \to \infty .
\end{equation*}
where $\widehat{\B}^t_{{k_t}}$ is any basis of $\spn\{\widehat{\V}_{k_t}^t\}^\perp$ with 
\begin{align*}
    \widehat{\V}_{k_t}^t= \argmin_{\V \in \spc(p,q)}L_{n,\Fa}^*(\V,f_t).
\end{align*}
with $q = p - k_t$ and $k_t = \dim(\msf)$.
\end{thm}

A straightforward application of Theorem \ref{thm_consistency_mean_subspace}, using the identity function, obtains that  $\spc_{E(Y\mid \X)}$ can be consistently estimated by ECVE. 

\begin{thm}\label{uniform_convergence_eobjective}
Assume the conditions of Theorem~\ref{uniform_convergence_ecve} hold. Let $\Fa$ be a parametric ensemble such that $\sup_{t \in \Omega_T} |f_t(Y)| < M < \infty$ almost surely, and let the index random variable $t \sim F_T$  be independent from the data $(Y_i, \X_i)_{i=1,\ldots,n}$. Then $L_{n,\Fa}(\V)$, defined in \eqref{e_objective_est}, converges uniformly in probability to $L_{\Fa}(\V)$  in \eqref{e_objective}; i.e.,
\begin{align*}
    \sup_{\V \in \spc(p,q)}|L_{n,\Fa}(\V) -L_{\Fa}(\V)| \longrightarrow 0 \quad  \text{in probability as $n \to \infty$.}
\end{align*}
\end{thm}

The assumption $\sup_{t \in \Omega_T} |f_t(Y)| < M < \infty$ in Theorem~\ref{uniform_convergence_eobjective}  
is trivially satisfied by the elements of  the characteristic and indicator ensembles. Further the assumption $a_n/h_n^{(p-q)/2} = O(1)$ used for the truncation step in the proof of Theorem~\ref{uniform_convergence_ecve} can be dropped since obviously no truncation is needed.

The rate of convergence of $m_n$ is not characterized in Theorem~\ref{uniform_convergence_eobjective}. In the simulation studies of Sections \ref{sec:consistency_simulations} and \ref{sec:simulations}, we find that $m_n$ should be chosen to be very small relative to the sample size $n$, roughly at the rate of $\log(n)$.

The consistency of the ensemble CVE is shown in Theorem~\ref{ECVE_consistency}. 

\begin{thm}\label{ECVE_consistency}
Assume the conditions of Theorem~\ref{uniform_convergence_ecve} and (E.3) hold. Let $\Fa$ be a parametric ensemble that characterizes $\cs$ and whose members satisfy  $\sup_{t \in \Omega_T} |f_t(Y)| < M < \infty$ almost surely. Also, assume the index random variable $t \sim F_T$  is independent from the data $(Y_i, \X_i)_{i=1,\ldots,n}$. Then, the \textit{ensemble conditional variance estimator (ECVE)} is a consistent estimator for $\cs$. That is, for any basis $\widehat{\B}_{p-q, \Fa}$  of $\spn\{\widehat{\V}_q\}^\perp$, where $\widehat{\V}_q$ is defined in \eqref{e_optim} with $q = p- k$ and $k = \dim(\cs)$, 
\begin{align*}
    \|\Pbf_{\widehat{\B}_{p-q, \Fa}} - \Pbf_{\cs}\| \longrightarrow 0 \quad \text{in probability as } n \to \infty,
\end{align*}
where $\Pbf_{\M}$ denotes the orthogonal projection onto the range space of the matrix or linear subspace $\M$.
\end{thm}

\section{Simulation Studies}\label{sec:simulations}
\subsection{Influence of $m_n$ on ECVE}\label{sec:influence}

In this section we study the influence of the number of functions of the ensemble $\Fa$, $m_n$ in \eqref{e_objective_est}, on the accuracy of the ensemble conditional variance estimation. In Theorem~ \ref{uniform_convergence_eobjective} and \ref{ECVE_consistency}, how fast  $m_n$ approaches $\infty$ is unspecified. We consider the 2-dimensional regression model
\begin{equation}\label{ecve_simmodel}
    Y = (\bb_2^T\X) + (0.5 + (\bb_1^T\X)^2)\epsilon,
\end{equation}
where  $p = 10$, $k =2$, $\X \sim N(0,\I_{10})$, $\epsilon \sim N(0,1)$ independent of $\X$, 
$\bb_1 = (1,0,\ldots,0)^T \in \real^{p}$, and $\bb_2=(0,1,0,\ldots,0)^T \in \real^{p}$. Therefore, $\ms = \spn\{\bb_2\} \subsetneq \cs = \spn\{\B\}$, with $\B = (\bb_1,\bb_2)$. 

We set the sample size to $n = 300$ and  vary $m$ over $\{4,8,10,26,50,76,100\}$ for the (a) indicator, $\Fa_{m,\text{Indicator}}= \{1_{\{x \geq q_j\}}: j = 1,\ldots,m\}$, where $q_j$ is the $j/(m +1)$th empirical quantile of $(Y_i)_{i=1,\ldots,n}$; (b) characteristic or Fourier, $\Fa_{m,\text{Fourier}}= \{\sin(jx): j = 1,\ldots,m/2\} \cup \{\cos(jx): j = 1,\ldots,m/2\}$; (c) monomial, $\Fa_{m,\text{Monom}}= \{x^j: j = 1,\ldots,m\}$, (d) and Box-Cox, $\Fa_{m,\text{BoxCox}}= \{(x^{t_j} -1)/t_j: t_j = 0.1 + 2(j-1)/(m-1), j = 1,\ldots,m-1\} \cup \{\log(x)\}$, ensembles. 

For each ensemble, we form the ensemble conditional variance estimator and its weighted version as in section~\ref{weight_section}, 
see also \cite{FertlBura}.  The results of 100 replications for each method and each $m$ are displayed in Figure~\ref{ecve_test_nr basisfunction_Fig2}. We assess the estimation accuracy with $\text{err}_{j,m} = \|\widehat{\B}\widehat{\B}^T - \B\B^T\|/(2k)^{1/2}$, $j=1,\ldots,100$, $m \in \{2,4,8,10,26,50,76,100\}$. \texttt{ECVE}'s main competitor, \texttt{csMAVE}, which does not vary with $m$, estimate of the central subspace has median error 0.2 with a wide range from 0.1 to 0.6. The estimation accuracy of Fourier, Indicator and Box-Cox \texttt{ECVE} vary over $m$ and is on par or better for some $m$ values. 

For the Fourier basis, fewer basis functions give the best performance, the indicator and BoxCox ensembles are quite robust against varying $m$, whereas the errors get rapidly larger if $m$ is increased for the monomial ensemble. The weighted version of ECVE improves the accuracy for all ensembles.  $\Fa_{4,\text{Fourier\_weighted}}$, $\Fa_{8,\text{Indicator\_weighted}}$, $\Fa_{4,\text{BoxCox\_weighted}}$ are on par or more accurate than \texttt{csMAVE}. 
In sum, the simulation results support a choice of a small $m$ number of basis functions. Based on this and further unreported simulations, we set the default value of $m$ to
\begin{align}\label{mn_value}
    m_n = \begin{cases}
\lceil \log(n)\rceil,\text{if}\quad\lceil \log(n)\rceil \quad  \text{even}\\
\lceil \log(n)\rceil + 1 ,\text{if}\quad \lceil \log(n)\rceil \quad  \text{odd}\\
\end{cases}
\end{align}
for all simulations in  Section~\ref{sec:consistency_simulations}, \ref{evaluate_acc} and the data analysis in Section~\ref{sec:dataAnalysis}.

\subsection{Demonstrating consistency}\label{sec:consistency_simulations}

We explore the consistency rate of the \textit{conditional variance estimator (CVE)} and \textit{ensemble conditional variance estimator (CVE)}, \texttt{csMAVE} and \texttt{mMAVE} in model~\eqref{ecve_simmodel}.

Specifically, we apply seven estimation methods,  the first five targeting the central subspace $\cs$ and the last two $\ms$, as follows.  
For $\cs$, we compare \texttt{ECVE} for the indicator (I), Fourier (II), monomial (III) and Box-Cox (IV) ensembles, as in Section \ref{sec:influence}, and \texttt{csMAVE} (V). 
For $\ms$, we use \texttt{CVE} (VI) of \cite{FertlBura} and \texttt{mMAVE} (VII) in \cite{Xiaetal2002}.

The simulation is performed as follows. 
We generate 100 i.i.d samples $(Y_i,\X_i^T)_{i=1,...,n}$ from \eqref{ecve_simmodel} for each sample size $n=100,200,400,600,800,1000$. Model \eqref{ecve_simmodel} is a two dimensional model with $\ms=\spn(\bb_2) \subsetneq \cs=\spn(\B)$. For methods (I)-(V), we set $k =2$ and estimate $\B \in \real^{10 \times 2}$. For (VI) and (VII),  we set $k=1$ and estimate $\bb_2 \in \real^{10 \times 1}$. Then, we calculate $\text{err}_{j,n} = \|\widehat{\B}\widehat{\B}^T - \B\B^T\|/(2k)^{1/2}$, $j=1,\ldots,100$, $n \in \{100,200,400,600,800,1000\}$. 
 Figure~\ref{ecve_consisteny_Fig1} displays the distribution of $err_{j,n}$ for increasing $n$ for the seven methods. As the sample size increases ECVE Indicator, Fourier and csMAVE are on par with respect to both speed and accuracy. The accuracy of ECVE Box-Cox improves as the sample size increases but at a slower rate. There is no improvement in the accuracy of ECVE monomial. This is not surprising as the monomial, as well as the Box-Cox, do not satisfy the assumption  $\sup_{t \in \Omega_T} |f_t(Y)| < M < \infty$ in Theorem~\ref{ECVE_consistency}, in contrast to the Indicator and Fourier ensembles. The Fourier, Indicator ECVE and csMAVE estimate $\cs = \spn\{\B\}$ consistently and the mean subspace methods, CVE and mMAVE, estimate $\ms = \spn\{\bb_2\}$ consistently.

\begin{figure}
\begin{minipage}{0.5\linewidth}
\centering
\includegraphics[width=1.04\textwidth]{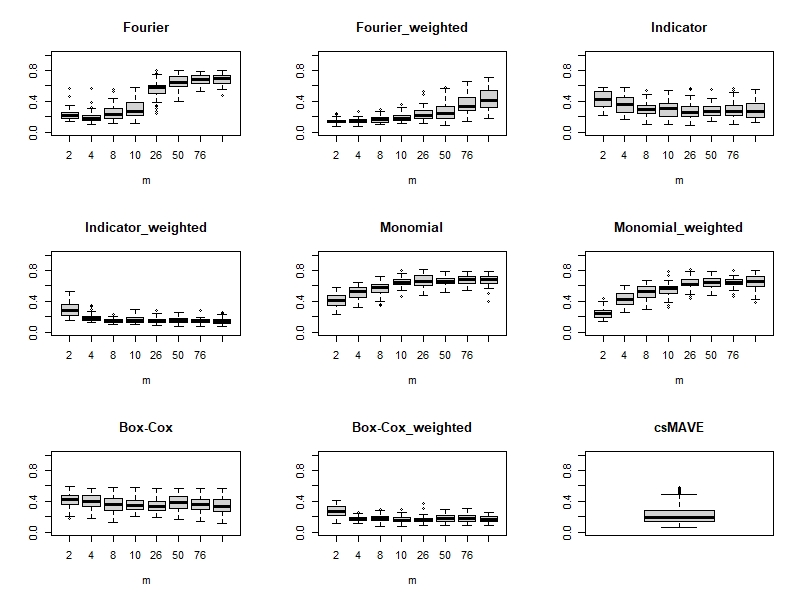}
		\caption{Box plots of the estimation errors over $100$ replications of model~\eqref{ecve_simmodel} with $n =300$ over $ m =|\Fa| = (2,4,8,10,26,50,76,100)$ across four ensembles.}
		\label{ecve_test_nr basisfunction_Fig2}
\end{minipage}
\quad
\begin{minipage}{0.5\linewidth}
\centering
\includegraphics[width=1.05\textwidth]{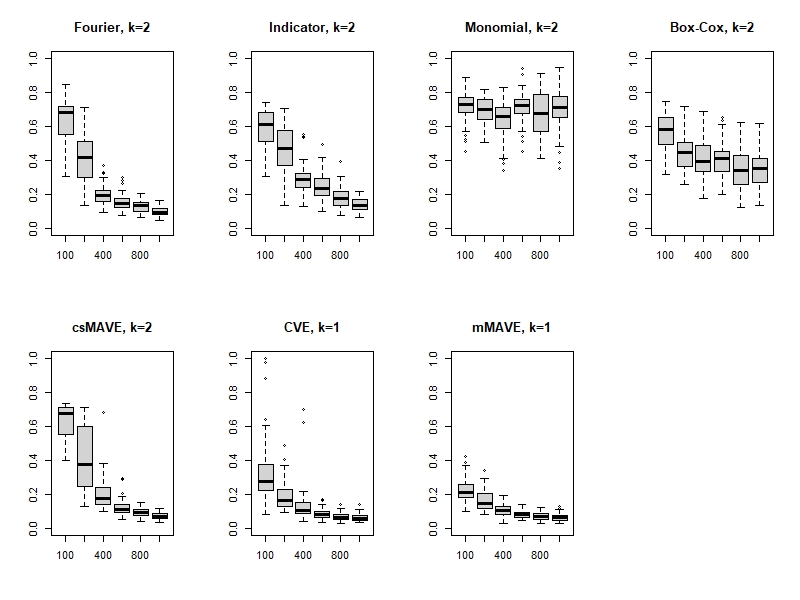}
		\caption{Estimation error distribution of model~\eqref{ecve_simmodel} plotted over $n = (100,200,400,600,800,1000)$ for the seven (I-VII) methods}
		\label{ecve_consisteny_Fig1}
\end{minipage}
\end{figure}

\subsection{Evaluating estimation accuracy}\label{evaluate_acc}
We consider seven models, (M1-M7) defined in Table~\ref{tab:e_mod}, three different sample sizes $\{100,200,400\}$, and three different distributions of the predictor vector $\X = \Sigmabf^{1/2}\Z \in \real^p$, where $\Sigmabf=(\Sigma_{ij})_{i,j=1,\ldots,p}$, $\Sigma_{i,j} = 0.5^{|i-j|}$. Throughout, $p=10$, $\B$ are the first $k$ columns of $\I_p$, and $\epsilon \sim N(0,1)$ independent of $\X$. 
As in \cite{WangXia2008}, we consider three distributions for $\Z \in \real^p$: 
    (I) $N(0,\I_p)$,  
    (II) $p$-dimensional uniform distribution on $[-\sqrt{3},\sqrt{3}]^p$, i.e. all components of $\Z$ are independent and uniformly distributed
    , and (III) a mixture-distribution $N(0,\I_p)+\mubf$, where $\mubf = (\mu_1,\ldots,\mu_p)^T \in \real^p$ with $\mu_j=2$, $\mu_k=0$, for $k \ne j$, and $j$ is uniformly distributed on $\{1,\ldots,p\}$. 
 
The simple and weighted [see Section~\ref{weight_section}] \texttt{Fourier} 
 and \texttt{Indicator} 
 ensembles are used to form four \textit{ensemble conditional variance estimators} (ECVE).
 The monomial and 
 BoxCox ensembles 
  were also used but did not give satisfactory results and are not reported.
From these two ensembles four ECVE estimators are formed and compared against the reference method \texttt{csMAVE} \cite{WangXia2008}, which is implemented in the \texttt{R} package \texttt{MAVE}. The source code for \textit{conditional variance estimation} and its ensemble version is available at \url{https://git.art-ist.cc/daniel/CVE}.

\begin{table}[!htbp]
\centering
\caption{Models}
{\small
\begin{tabular}{lcccc}
		\toprule
Name & Model  & $\ms$& $\cs$&$k$ \\ \midrule
M1& $Y = \frac{1}{\bb_1^T\X}+0.2\epsilon$ &$\spn\{\bb_1\} $&$\spn\{\bb_1\} $& 1\\
M2& $Y = \cos(2\bb_1^T\X)+\cos(\bb_2^T\X)+0.2\epsilon$&$\spn\{\bb_1,\bb_2\} $&$\spn\{\bb_1, \bb_2\}$ &2\\
M3&$Y = (\bb_2^T\X) + (0.5 + (\bb_1^T\X)^2)\epsilon$&$\spn\{\bb_2\} $&$\spn\{\bb_1,\bb_2\} $& 2\\
M4& $Y = \frac{\bb_1^T\X}{0.5+(1.5+\bb_2^T\X)^2}+(|\bb_1^T\X| + (\bb_2^T\X)^2 +0.5)\epsilon$&$\spn\{\bb_1,\bb_2\} $&$\spn\{\bb_1,\bb_2\}$ & 2\\
M5& $Y = \bb_3^T\X +\sin(\bb_1^T \X(\bb_2^T\X)^2)\epsilon$&$\spn\{\bb_3\} $&$\spn\{\bb_1,\bb_2,\bb_3\}$ & 3\\
M6& $Y = 0.5(\bb_1^T\X)^2\epsilon$&$\spn\{\0\} $&$\spn\{\bb_1\} $& 1\\
M7& $Y = \cos(\bb_1^T\X - \pi)+ \cos(2\bb_1^T\X)\epsilon$&$\spn\{\bb_1\} $&$\spn\{\bb_1\} $& 1\\
\bottomrule
	\end{tabular}%
	}
	\label{tab:e_mod}%
\end{table}

We set $q = p - k$ and generate $r=100$ replicates of models M1-M7 with the specified distribution of $\X$ and sample size $n$. We estimate $\B$ using the four
ECVE methods and csMAVE. The accuracy of the estimates is assessed using  $err= \|\Pbf_\B - \Pbf_{\widehat{\B}}\|_2/\sqrt{2k} \in [0,1]$, where $\Pbf_\B = \B(\B^T\B)^{-1}\B^T$ is the orthogonal projection matrix on $\spn\{\B\}$. The factor $\sqrt{2k}$ normalizes the distance, with values closer to zero indicating better agreement and values closer to one indicating strong disagreement.
The results are displayed in Tables~\ref{tab:summaryM1}-\ref{tab:summaryM8}. In M1, which is taken from \cite{WangXia2008}, the mean subspace agrees with the central subspace, i.e. $\ms =\cs$, but due to the unboundedness of the link function $g(x) = 1/x$ most mean subspace estimation methods, such as \texttt{SIR, mean MAVE} and \texttt{CVE}, fail. In contrast, all 4 ensemble CVE methods and  \texttt{csMAVE}  succeed in identifying the minimal dimension reduction subspace, with ensemble CVE performing slightly better, as can be seen in Table~\ref{tab:summaryM1}.   In particular,   \texttt{Fourier}  is the best performing method.
M2, is a two dimensional mean subspace model, i.e. $\ms = \cs$, and in Table~\ref{tab:summaryM2} we see that \texttt{csMAVE} is the best performing method. M3 is the same as model \eqref{ecve_simmodel} and here the mean subspace is a proper subset of the central subspace. In Table~\ref{tab:summaryM3} we see that \texttt{Indicator\_weighted} and \texttt{csMAVE} are the best performers and are roughly on par.
In M4, the two dimensional mean subspace, which determines also the heteroskedasticity, agrees with the central subspace. In Table~\ref{tab:summaryM4} we see that this model is quite challenging for all methods, and only \texttt{Indicator\_weighted} and \texttt{csMAVE} give satisfactory results, with \texttt{Indicator\_weighted} the clear winner. 

In M5, the heteroskedasticity is induced by an interaction term, and the three dimensional central subspace model is a proper superset of the one dimensional mean subspace. In Table~\ref{tab:summaryM5} we see that M5 is quite challenging for all five methods, therefore we increase the sample size $n$ to $800$. For M5, the two weighted ensemble conditional variance estimators are the best performing methods followed by \texttt{csMAVE}.

M6 is a one dimensional pure central subspace model, whereas the mean subspace is $0$. In Table~\ref{tab:summaryM7}, we see that for $n = 100$ the two weighted ECVEs are the best performing methods and for higher sample sizes \texttt{csMAVE} is slightly more accurate than the ECVE methods.

In M7 the one dimensional mean subspace agrees with the central subspace, i.e. $\ms = \cs$, and the conditional first and second moments, $\E(Y^l \mid \X)$ for $l =1,2$, are highly  nonlinear and periodic functions of the sufficient reduction. In Table~\ref{tab:summaryM8}, we see that all ensemble conditional variance estimators clearly outperform  \texttt{csMAVE}. 

 \begin{table}[!htbp]
\centering
\caption{Mean and standard deviation  (in parenthesis) of estimation errors of M1}
{\tiny
\begin{tabular}{ll|ccccc}
		\toprule
Distribution &$n$ &\texttt{Fourier} & \texttt{Fourier\_weighted} & \texttt{Indicator}& \texttt{Indicator\_weighted}& \texttt{csMAVE}  \\ \midrule
I& 100& \textbf{0.172} &0.201 &0.248& 0.265 & 0.210\\
& & (0.047) &(0.054)& (0.064)& (0.063)&  (0.063)\\
\midrule
I& 200&  \textbf{0.120} &0.142& 0.182& 0.197&  0.128\\
& & (0.029)& (0.037) &(0.045)& (0.049)&  (0.037)\\
\midrule
I& 400& \textbf{0.079}& 0.091 &0.126& 0.136& 0.080\\
& & (0.020)& (0.024) &(0.037)& (0.040)& (0.024)\\
\midrule
II& 100& \textbf{0.174}&0.196& 0.241& 0.254&{0.193}\\
& & (0.038)& (0.049) &(0.055)& (0.056)& (0.059)\\
\midrule
II& 200& \textbf{0.110} &0.127& 0.170& 0.182&  {0.121}\\
& & (0.031)& (0.033) &(0.043)& (0.045)&  (0.036)\\
\midrule
II& 400& \textbf{0.078}& 0.091 &0.122 &0.132& {0.079}\\
& & (0.021) &(0.026)& (0.031)& (0.033) &(0.020)\\
\midrule
III& 100& \textbf{0.187}& 0.218& 0.256 &0.263& {0.204}\\
& &(0.045)& (0.053)& (0.060)& (0.058) &(0.066)\\
\midrule
III& 200& \textbf{0.118}& 0.137 &0.171&0.179 & \textbf{0.118}\\
& & (0.031)& (0.038)& (0.043)& (0.042)& (0.033)\\
\midrule
III& 400& 0.082 &0.101&0.127 &0.132& \textbf{0.079}\\
& & (0.020) &(0.029) &(0.031) &(0.032) &(0.022)\\
 		\bottomrule
	\end{tabular}%
	}
	\label{tab:summaryM1}%
\end{table}

 \begin{table}[!htbp]
\centering
\caption{Mean and standard deviation  (in parenthesis) of estimation errors of M2}
{\tiny
\begin{tabular}{ll|ccccc}
		\toprule
Distribution &$n$ &\texttt{Fourier} & \texttt{Fourier\_weighted} & \texttt{Indicator}& \texttt{Indicator\_weighted}& \texttt{csMAVE}  \\ \midrule
I& 100& {0.670} &0.601 &0.629& 0.582 & \textbf{0.575}\\
& & (0.089) &(0.135)& (0.130)& (0.140)&  (0.176)\\
\midrule
I& 200&  {0.478} &0.388& 0.436& 0.407&  \textbf{0.219}\\
& & (0.201)& (0.152) &(0.193)& (0.162)&  (0.136)\\
\midrule
I& 400& {0.226}& 0.201 &0.231& 0.236& \textbf{0.098}\\
& & (0.153)& (0.074) &(0.127)& (0.111)& (0.025)\\
\midrule
II& 100&{0.663}&0.652& 0.687& 0.658&\textbf{0.544}\\
& & (0.097)& (0.104) &(0.057)& (0.080)& (0.176)\\
\midrule
II& 200& {0.525} &0.468& 0.601& 0.539&  \textbf{0.182}\\
& & (0.171)& (0.171) &(0.127)& (0.148)&  (0.096)\\
\midrule
II& 400& {0.267}& 0.307 &0.375 &0.357& \textbf{0.087}\\
& & (0.081) &(0.146)& (0.154)& (0.141) &(0.021)\\
\midrule
III& 100& {0.657}& 0.590&\textbf{0.530} &0.542& {0.603}\\
& &(0.104)& (0.148)& (0.155)& (0.148) &(0.193)\\
\midrule
III& 200& {0.421}& 0.367 &0.306&0.336& \textbf{0.240}\\
& & (0.203)& (0.165)& (0.147)& (0.151)& (0.193)\\
\midrule
III& 400& 0.170&0.170&0.144&0.170& \textbf{0.089}\\
& & (0.110) &(0.071) &(0.053) &(0.063) &(0.019)\\
 		\bottomrule
	\end{tabular}%
	}
	\label{tab:summaryM2}%
\end{table}
 \begin{table}[!htbp]
\centering
\caption{Mean and standard deviation  (in parenthesis) of estimation errors of M3}
{\tiny
\begin{tabular}{ll|ccccc}
		\toprule
Distribution &$n$ &\texttt{Fourier} & \texttt{Fourier\_weighted} & \texttt{Indicator}& \texttt{Indicator\_weighted}& \texttt{csMAVE}  \\ \midrule
I& 100& 0.744& 0.657& 0.668& \textbf{0.561}& 0.602\\
& & (0.056)&(0.113)&(0.083)&(0.142)&(0.147)\\
\midrule
I& 200&0.702 &0.472 &0.559& \textbf{0.369}& 0.374 \\
& & (0.061)&(0.177)&(0.147)&(0.155)&(0.148)\\
\midrule
I&  400& 0.621& 0.252& 0.408& 0.223& \textbf{0.203}\\
& & (0.148)&(0.102)&(0.177)&(0.064)&(0.061)\\
\midrule
II& 100&0.751& 0.698& 0.683& \textbf{0.570} &0.635\\
& & (0.041)&(0.076)&(0.080)&(0.136)&(0.136)\\
\midrule
II& 200& 0.719& 0.521& 0.584 &\textbf{0.355} &0.387\\
& & (0.040)&(0.163)&(0.111)&(0.097)&(0.144)\\
\midrule
II&  400& 0.686& 0.267& 0.452 &0.252& \textbf{0.201}\\
& & (0.079)&(0.084)&(0.153)&(0.052)&(0.045)\\
\midrule
III& 100& 0.739& 0.676& 0.654& \textbf{0.563}&0.571\\
& &(0.073)&(0.106)&(0.105)&(0.150)&(0.120)\\
\midrule
III& 200& 0.704& 0.546& 0.523& 0.368&\textbf{0.330}\\
& & (0.048)&(0.162)&(0.171)&(0.153)&(0.131)\\
\midrule
III& 400& 0.616 &0.252 &0.297& 0.202 &\textbf{0.179}\\
& & (0.151)&(0.113)&(0.106)&(0.055)&(0.042)\\
 		\bottomrule
	\end{tabular}%
	}
	\label{tab:summaryM3}%
\end{table}
 \begin{table}[!htbp]
\centering
\caption{Mean and standard deviation  (in parenthesis) of estimation errors of M4}
{\tiny
\begin{tabular}{ll|ccccc}
		\toprule
Distribution &$n$ &\texttt{Fourier} & \texttt{Fourier\_weighted} & \texttt{Indicator}& \texttt{Indicator\_weighted}& \texttt{csMAVE}  \\ \midrule
I& 100& {0.836} &0.794 &0.774& \textbf{0.713} & {0.803}\\
& & (0.072) &(0.076)& (0.074)& (0.105)&  (0.087)\\
\midrule
I& 200&  {0.820} &0.733& 0.747& \textbf{0.545}&  {0.685}\\
& & (0.066)& (0.094) &(0.060)& (0.150)&  (0.116)\\
\midrule
I& 400& {0.782}& 0.633 &0.710& \textbf{0.364}& {0.534}\\
& & (0.059)& ( 0.142) &(0.081)& (0.129)& (0.155)\\
\midrule
II& 100&{0.839}&0.828& 0.788& \textbf{0.751}&{0.818}\\
& & (0.067)& (0.064) &(0.062)& (0.095)& (0.095)\\
\midrule
II& 200& {0.834} &0.781& 0.759& \textbf{0.660}&  {0.701}\\
& & (0.171)& (0.081) &(0.040)& (0.117)&  (0.111)\\
\midrule
II& 400& {0.812}& 0.712 &0.739 &\textbf{0.511}& {0.544}\\
& & (0.059) &(0.097)& (0.038)& (0.135) &(0.151)\\
\midrule
III& 100& {0.838}& 0.815&{0.764} &\textbf{0.706}& {0.786}\\
& &(0.074)& (0.077)& (0.069)& (0.108) &(0.109)\\
\midrule
III& 200& {0.829}&0.761&0.726&\textbf{0.544}& {0.676}\\
& & (0.071)& (0.099)& (0.083)& (0.149)& (0.123)\\
\midrule
III& 400& 0.796&0.646&0.669&\textbf{0.317}& {0.506}\\
& & (0.069) &(0.139) &(0.113) &(0.110) &(0.146)\\
 		\bottomrule
	\end{tabular}%
	}
	\label{tab:summaryM4}%
\end{table}

 \begin{table}[!htbp]
\centering
\caption{Mean and standard deviation  (in parenthesis) of estimation errors of M5}
{\tiny
\begin{tabular}{ll|ccccc}
		\toprule
Distribution &$n$ &\texttt{Fourier} & \texttt{Fourier\_weighted} & \texttt{Indicator}& \texttt{Indicator\_weighted}& \texttt{csMAVE}  \\ \midrule
I& 100& 0.705& \textbf{0.682}& 0.708 &0.691 &0.709\\
& & (0.060)&(0.067)&(0.060)&(0.056)&(0.069)\\
\midrule
I& 200&  0.679 &\textbf{0.634} &0.688 &0.642 &0.687\\
& & (0.061)&(0.054)&(0.058)&(0.060)&(0.073)\\
\midrule
I& 400& 0.644 &\textbf{0.588}& 0.660& 0.591& 0.646\\
& & (0.050)&(0.047)&(0.056)&(0.061)&(0.082)\\
\midrule
I& 800& 0.622& 0.543& 0.629& \textbf{0.493}& 0.553\\
& & (0.032)&(0.078)&(0.035)&(0.100)&(0.077)\\
\midrule
II& 100&0.712& \textbf{0.688}& 0.713& 0.697 &0.722\\
& & (0.060)&(0.069)&(0.051)&(0.057)&(0.054)\\
\midrule
II& 200& 0.693& \textbf{0.669} &0.694&\textbf{0.669}& 0.697\\
& & (0.058)&(0.065)&(0.054)&(0.057)&(0.064)\\
\midrule
II& 400& 0.670& \textbf{0.614}& 0.681& 0.633& 0.687\\
& & (0.054)&(0.059)&(0.052)&(0.050)&(0.067)\\
\midrule
II& 800& 0.660& \textbf{0.584}& 0.672& 0.585& 0.589\\
& & (0.053)&(0.045)&(0.052)&(0.055)&(0.074)\\
\midrule
III& 100& 0.706& \textbf{0.687}& 0.703& 0.691 &0.724\\
& &(0.062)&(0.062)&(0.061)&(0.061)&(0.051)\\
\midrule
III& 200& 0.701& \textbf{0.655} &0.702& 0.668& 0.703\\
& & (0.063)&(0.069)&(0.058)&(0.074)&(0.080)\\
\midrule
III& 400& 0.659 &\textbf{0.603} &0.664 &0.604& 0.682\\
& & (0.062)&(0.072)&(0.059)&(0.077)&(0.081)\\
\midrule
III& 800&0.657 &0.562& 0.651 &\textbf{0.513}& 0.602\\
& & (0.064)&(0.068)&(0.052)&(0.109)&(0.087)\\
 		\bottomrule
	\end{tabular}%
	}
	\label{tab:summaryM5}%
\end{table}
 \begin{table}[!htbp]
\centering
\caption{Mean and standard deviation  (in parenthesis) of estimation errors of M6}
{\tiny
\begin{tabular}{ll|ccccc}
		\toprule
Distribution &$n$ &\texttt{Fourier} & \texttt{Fourier\_weighted} & \texttt{Indicator}& \texttt{Indicator\_weighted}& \texttt{csMAVE}  \\ \midrule
I& 100& {0.304} &\textbf{0.294} &0.492& {0.299} & {0.539}\\
& & (0.092) &(0.082)& (0.135)& (0.087)&  (0.255)\\
\midrule
I& 200&  {0.217} &0.213& 0.329& {0.205}&  \textbf{0.194}\\
& & (0.057)& (0.054) &(0.107)& (0.059)&  (0.061)\\
\midrule
I& 400& {0.142}& 0.146 &0.199& {0.138}& \textbf{0.114}\\
& & (0.036)& ( 0.035) &(0.069)& (0.039)& (0.034)\\
\midrule
II& 100&{0.308}&\textbf{0.293}& 0.479& {0.299}&{0.488}\\
& & (0.094)& (0.073) &(0.129)& (0.086)& (0.248)\\
\midrule
II& 200& {0.205} &0.210& 0.321&{0.210}&  \textbf{0.192}\\
& & (0.058)& (0.057) &(0.095)& (0.058)&  (0.061)\\
\midrule
II& 400& {0.144}& 0.150 &0.190 &{0.142}& \textbf{0.111}\\
& & (0.039) &(0.042)& (0.055)& (0.045) &(0.032)\\
\midrule
III& 100& {0.373}& 0.375&{0.504} &\textbf{0.322}& {0.562}\\
& &(0.152)& (0.175)& (0.143)& (0.083) &(0.273)\\
\midrule
III& 200& {0.226}&0.230&0.340&\textbf{0.218}& \textbf{0.218}\\
& & (0.065)& (0.070)& (0.100)& (0.060)& (0.083)\\
\midrule
III& 400& 0.149&0.151&0.194&{0.146}& \textbf{0.114}\\
& & (0.039) &(0.038) &(0.068) &(0.042) &(0.032)\\
 		\bottomrule
	\end{tabular}%
	}
	\label{tab:summaryM7}%
\end{table}

 \begin{table}[!htbp]
\centering
\caption{Mean and standard deviation  (in parenthesis) of estimation errors of M7}
{\tiny
\begin{tabular}{ll|ccccc}
		\toprule
Distribution &$n$ &\texttt{Fourier} & \texttt{Fourier\_weighted} & \texttt{Indicator}& \texttt{Indicator\_weighted}& \texttt{csMAVE}  \\ \midrule
I& 100& {0.273} &\textbf{0.237} &0.241& {0.252} & {0.790}\\
& & (0.169) &(0.050)& (0.136)& (0.158)&  (0.316)\\
\midrule
I& 200&  {0.160} &0.159& \textbf{0.143}&  {0.153}& {0.425}\\
& & (0.093)& (0.041) &(0.083)& (0.093)&  (0.391)\\
\midrule
I& 400& {0.098}& 0.104 &\textbf{0.088}& {0.102}& {0.127}\\
& & (0.024)& ( 0.025) &(0.021)& (0.093)& (0.202)\\
\midrule
II& 100&\textbf{0.233}&{0.260}& 0.236& {0.265}&{0.902}\\
& & (0.057)& (0.134) &(0.142)& (0.185)& (0.219)\\
\midrule
II& 200& {0.154} &0.176& \textbf{0.145}&{0.150}&  {0.649}\\
& & (0.058)& (0.124) &(0.093)& (0.094)&  (0.414)\\
\midrule
II& 400& {0.097}& 0.110 &\textbf{0.087} &{0.099}& {0.295}\\
& & (0.025) &(0.094)& (0.022)& (0.093) &(0.391)\\
\midrule
III& 100& {0.274}& 0.303&\textbf{0.238} &{0.298}& {0.933}\\
& &(0.201)& (0.237)& (0.160)& (0.242) &(0.163)\\
\midrule
III& 200& {0.167}&0.188&\textbf{0.159}&{0.167}& {0.678}\\
& & (0.120)& (0.159)& (0.150)& (0.144)& (0.408)\\
\midrule
III& 400& 0.100&0.116&\textbf{0.089}&{0.112}& {0.375}\\
& & (0.023) &(0.090) &(0.023) &(0.129) &(0.431)\\
 		\bottomrule
	\end{tabular}%
	}
	\label{tab:summaryM8}%
\end{table}

\section{Boston Housing Data}\label{sec:dataAnalysis}
We apply the ensemble conditional variance estimator and \texttt{csMAVE} to the \texttt{Boston Housing} data set. This data set has been extensively used as a benchmark for assessing regression methods [see, for example, \cite{James2013}], and is available in the \texttt{R}-package \texttt{mlbench}. 
The data contains 506 instances of 14 variables from the  1970 Boston census, 13 of which are continuous. The binary variable \texttt{chas}, indexing proximity to the Charles river, is omitted from the analysis since ensemble conditional variance estimation operates under the assumption of continuous predictors. The target variable is the median value of owner-occupied homes, \texttt{medv}, in \$1,000. The 12 predictors are \texttt{crim} (per capita crime rate by town), \texttt{zn} (proportion of residential land zoned for lots over 25,000 sq.ft), \texttt{indus} 	(proportion of non-retail business acres per town),
\texttt{nox} (nitric oxides concentration (parts per 10 million)), \texttt{rm} (average number of rooms per dwelling), \texttt{age} (proportion of owner-occupied units built prior to 1940), \texttt{dis} (weighted distances to five Boston employment centres), \texttt{rad} (index of accessibility to radial highways), \texttt{tax} (full-value property-tax rate per \$10,000), \texttt{ptratio} (pupil-teacher ratio by town),  \texttt{lstat} (percentage of lower status of the population), and \texttt{b} stands for $1000(B - 0.63)^2$ where $B$ is the proportion of blacks by town. 

\medskip
We analyze these data with the weighted and unweighted Fourier and indicator ensembles, and \texttt{csMAVE}. We compute  unbiased error estimates by leave-one-out  
cross-validation. 
We estimate the sufficient reduction with the five methods from the standardized training set, estimate the forward model from the reduced training set using \texttt{mars}, multivariate adaptive regression splines \cite{mars}, in the \texttt{R}-package \texttt{mda}, and predict the target variable on the test set. 
We report results for dimension $k =1$. 
The analysis was  repeated setting $k = 2$ with similar results.  
Table~\ref{table:datasets} reports the first quantile, median, mean and third quantile of the out-of-sample prediction errors. The reductions estimated by the ensemble \texttt{CVE} methods achieve lower mean and median prediction errors than \texttt{csMAVE}. Also, both \texttt{ensemble CVE} and \texttt{csMAVE} are approximately on par with the variable selection methods in \cite[Section 8.3.3]{James2013}.

\begin{table}[!htbp]
\centering
\caption{Summary statistics of the out of sample prediction errors for the Boston Housing data obtained by LOO cross validation}
\vspace{0.05in}
{\footnotesize
\begin{tabular}{l|ccccc}
		\toprule
& \texttt{Fourier} & \texttt{Fourier\_weighted} & \texttt{Indicator}& \texttt{Indicator\_weighted}& \texttt{csMAVE}  \\ 
\bottomrule
25\% quantile & \textbf{0.766}  &  0.785  &0.973& 0.916&0.851\\
median  & \textbf{3.323}  &  3.358 &3.844& 3.666&4.515\\
mean & 19.971  & 19.948 &19.716& \textbf{19.583} &24.309\\
75\% quantile & 11.129     & 10.660  &11.099& \textbf{10.429} &16.521\\
\bottomrule

	\end{tabular}%
	}
	\label{table:datasets}%
\end{table}

\medskip
Moreover, we plot the standardized response \texttt{medv} against the reduced \texttt{Fourier} and \texttt{csMAVE} predictors, $\B^T\X$, in Figure~\ref{Fig_real_data}. The sufficient reductions are estimated using the entire data set. A particular feature of these data is that the response \texttt{medv} appears to be truncated  as the highest median price of exactly \$50,000 is reported in 16 cases. Both methods pick up similar patterns, which is captured by the relatively high absolute correlation of the coefficients of the two reductions, $|\widehat{\B}_{\texttt{Fourier}}^T \widehat{\B}_{\texttt{csMAVE}}| = 0.786$. The coefficients of the reductions, $\widehat{\B}_{\texttt{Fourier}}$ and $\widehat{\B}_{\texttt{csMAVE}}$, are reported in Table~\ref{tab:dataset_reductions}. For the \texttt{Fourier} ensemble, the variables \texttt{rm} and \texttt{lstat} have the highest influence on the target variable \texttt{medv}. This agrees with the analysis in \cite[Section 8.3.4]{James2013} where it was found that these two variables are by far the most important using different variable selection techniques, such as random forests and boosted regression trees. In contrast, the reduction estimated by \texttt{csMAVE} has a lower coefficient for \texttt{rm} and higher ones for \texttt{crim} and \texttt{rad}.

\begin{table}[!htbp]
{\scriptsize
\centering
\caption{{\footnotesize Rounded coefficients of the estimated reductions for $\widehat{\B}_{\texttt{Fourier}}$ and $\widehat{\B}_{\texttt{csMAVE}}$ from the full Boston Housing data}}
\vspace{0.05in}
\begin{tabular}{l|cccccccccccc}
		\toprule
& \texttt{crim}&    \texttt{zn}&      \texttt{indus}&  \texttt{ nox}&     \texttt{rm}&      \texttt{age}&\texttt{dis}&\texttt{rad}&\texttt{tax}&\texttt{pt}\texttt{ratio}& \texttt{b}&\texttt{lstat}\\ 
\bottomrule
\texttt{Fourier}&0.21 &-0.01 &0.04 & 0.1 &-0.62 &0.16&  0.2 &   0 & 0.2 & 0.27 &-0.25 & 0.57\\
\texttt{csMAVE}&0.5& -0.05 &-0.06& 0.14& -0.27& 0.11& 0.24& -0.43& 0.3&    0.19& -0.15 & 0.51\\
\bottomrule
\end{tabular}%
	\label{tab:dataset_reductions}%
}
\end{table}

\begin{figure}[htbp] 
	\centering
	\includegraphics[width=0.6\textwidth]{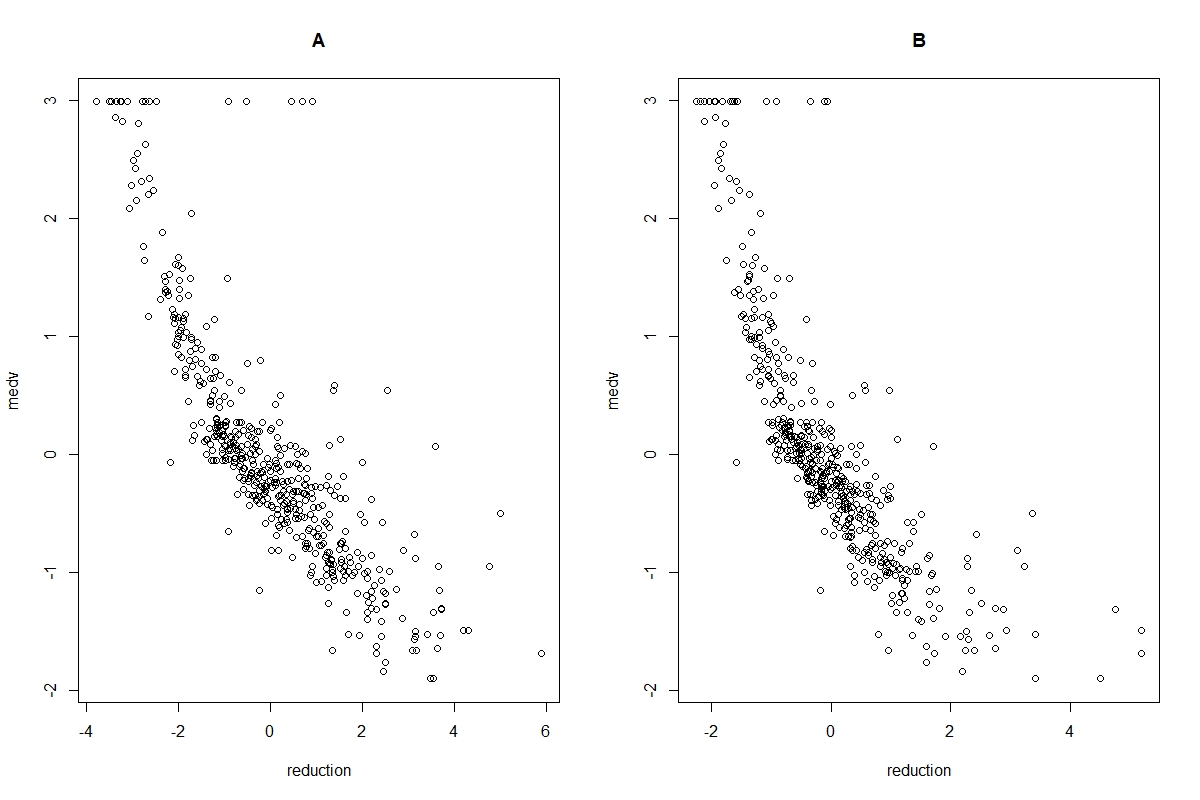}
	\caption{ Panel A: $Y$ vs. $\widehat{\B}_{\texttt{Fourier}}^T\X$. Panel B: $Y$ vs. $\widehat{\B}_{\texttt{csMAVE}}^T\X$}
	\label{Fig_real_data}
\end{figure}


\section{Discussion}\label{sec:discussion}
In this paper, we extend the \textit{mean subspace} conditional variance estimation (\texttt{CVE}) of \cite{FertlBura} 
to the ensemble conditional variance estimation (\texttt{ECVE}), which exhaustively estimates the \textit{central subspace}, by applying the ensemble device introduced by \cite{YinLi2011}. In Section~\ref{sec:consistency} we showed that the new estimator is consistent for the central subspace. 
The regularity conditions for consistency require the joint distribution of the target variable and predictors, $(Y,\X^T)^T$, be sufficiently smooth. They are comparable to those under which the main competitor \texttt{csMAVE} \cite{WangXia2008} is consistent. 

We analysed the estimation accuracy of \texttt{ECVE}
in Section~\ref{sec:simulations}. We found  that it is either on par with \texttt{csMAVE} or that it exhibits substantial performance improvement in certain models. We could not characterize the defining features of the models for which the  ensemble conditional variance estimation  outperforms \texttt{csMAVE}. This is an interesting line of further research together with establishing more theoretical results such as the rate of convergence, estimation of the structural dimension, and the limiting distribution of the estimator.   

\texttt{ECVE} identifies the central subspace via the orthogonal complement and thus circumvents the estimation and inversion of the variance matrix of the predictors $\X$. This renders the method formally applicable to settings where the sample size $n$ is small or smaller than  $p$, the number of predictors, and leads to potential future research. 

Throughout, the dimension of the central subspace, $k = \dim(\cs)$, is assumed to be known. The derivation of asymptotic tests for dimension is technically very challenging due to the lack of closed-form solution and the lack of independence of all quantities in the calculation. The dimension can be estimated via cross-validation, as in \cite{WangXia2008} and \cite{FertlBura}, or information criteria. 

\section*{Acknowledgements}
The authors gratefully acknowledge the support of the Austrian Science Fund (FWF P 30690-N35) and thank Daniel Kapla for his programming assistance. 
Daniel Kapla also co-authored the \texttt{CVE} \texttt{R} package that implements the proposed method.


\newcommand{\etalchar}[1]{$^{#1}$}

\newpage
\section*{Appendix}
For any $\V \in \spc(p,q)$, defined in \eqref{Smanifold}, we generically denote a basis of the orthogonal complement of its column space $\spn\{\V\}$, by $\Ub$. That is, $\Ub \in \spc(p,p-q)$ such that $\spn\{\V\} \perp \spn\{\Ub\}$ and $\spn\{\V\} \cup \spn\{\Ub\} = \real^p$, $\Ub^T\V = \0 \in \real^{(p-q) \times q}, \Ub^T\Ub = \I_{p-q}$. For any $\xn, \bs_0 \in \real^p$ we can always write
\begin{equation}\label{ortho_decomp}
    \xn = \bs_0 + \Pbf_\V (\xn - \bs_0) + \Pbf_\Ub (\xn - \bs_0) = \bs_0 + \V\rs_1 + \Ub\rs_2
\end{equation}
where $\rs_1 = \V^T(\xn-\bs_0) \in \real^{q}, \rs_2 = \Ub^T(\xn-\bs_0) \in \real^{p-q}$.

\begin{proof}[Proof of Theorem~\ref{CVE_targets_meansubspace_thm}]
The density of $\X \mid \X \in \bs_0 + \spn\{\V\}$ is given by 
\begin{gather}\label{density}
f_{\X\mid\X \in \bs_0 +\spn\{\V\}}(\rs_1) =
\frac{f_\X(\bs_0 + \V\rs_1)}{\int_{\real^q}f_\X(\bs_0 + \V\rs)d\rs} 
\end{gather}
where $\X$ is the $p$-dimensional  continuous random covariate vector with density $f_\X(\xn)$, $\bs_0 \in \text{supp}(f_\X) \subset \real^p$, and $\V$ belongs to the Stiefel manifold $\spc(p,q)$ defined in \eqref{Smanifold}. Equation \eqref{density} follows from Theorem 3.1 of \cite{Leaoetal2004}  and the fact that  $(\real^p, \mathcal{B}(\real^p))$, where $\mathcal{B}(\real^p)$ denotes the Borel sets on $\real^p$, is a Polish space, which in turnguarantees the existence of the regular conditional probability of $\X\mid\X \in \bs_0 +\spn\{\V\}$  [see also \cite{regProb}]. Further, the measure is concentrated on the  affine subspace $\bs_0 +\spn\{\V\} \subset \real^p$ with density  \eqref{density}, which follows from Definition 8.38, Theorem 8.39 of \cite{regProb2}, the orthogonal decomposition \eqref{ortho_decomp}, and the continuity of $f_\X$ (E.2). 

By assumption (E.1), $Y = \gcs(\B^T\X,\epsilon)$ with $\eps \ind \X$. 
Assume $f \in \Fa$ for which assumption (E.4) holds and let  $\widetilde{\B}$ be a basis of $\msf$; that is, $\spn\{\widetilde{\B}\} = \msf \subseteq \cs = \spn\{\B\}$. By Theorem~\ref{Y_decomposition_thm},
 $   f(Y) = g(\widetilde{\B}^T\X) + \tilde{\eps}$, 
with $\E(\tilde{\eps}\mid \X) = 0$ and $g$  twice continuously differentiable. Therefore, 
\begin{gather}
    \tilde{L}_\Fa(\V, \bs_0,f) = \Var\left(f(Y)\mid \X \in \bs_0 + \spn\{\V\}\right) \notag\\
    = \var\left(g(\widetilde{\B}^T\X) \mid \X \in \bs_0 + \spn\{\V\} \right) +2\cov\left(\tilde{\eps},g(\widetilde{\B}^T\X) \mid \X \in \bs_0 + \spn\{\V\}\right)  \notag \\
    +\var\left(\tilde{\eps} \mid \X \in \bs_0 + \spn\{\V\} \right) \notag \\
    =\var\left(g(\widetilde{\B}^T\X) \mid \X \in \bs_0 + \spn\{\V\} \right) + \var\left(\tilde{\eps} \mid \X \in \bs_0 + \spn\{\V\} \right)\label{CVE_problem}
\end{gather}
The covariance term in \eqref{CVE_problem} vanishes since
\begin{gather*}
  \cov\left(\tilde{\eps},g(\widetilde{\B}^T\X) \mid \X \in \bs_0 + \spn\{\V\}\right) 
  = \E\left( \underbrace{\E(\tilde{\eps}\mid \X)}_{=0}g(\widetilde{\B}^T\X) \mid \X \in \bs_0 + \spn\{\V\}\right) \\
  - E\left(g(\widetilde{\B}^T\X) \mid \X \in \bs_0 + \spn\{\V\}\right)\E\left( \underbrace{\E(\tilde{\eps}\mid \X)}_{=0} \mid \X \in \bs_0 + \spn\{\V\}\right) = 0
\end{gather*}
i.e. the sigma field generated by $\X \in \bs_0 + \spn\{\V\}$ is a subset of that generated by $\X$. By the same argument and using \eqref{density}
\begin{gather}
    \var\left(\tilde{\eps}\mid \X \in \bs_0 + \spn\{\V\}\right) = \E(\tilde{\eps}^2\mid \X \in \bs_0 + \spn\{\V\}) \notag\\
    = \E(\E(\tilde{\eps}^2\mid \X) \mid \X \in \bs_0 + \spn\{\V\}) = \E(h(\X) \mid \X \in \bs_0 + \spn\{\V\}) \notag \\
    = \int_{\text{supp}(f_\X)\cap\real^q} h(\bs_0 + \V \rs_1 ) \times f_\X(\bs_0 + \V \rs_1)d\rs_1 / t^{(0)}(\V,\bs_0,f)) \notag
\end{gather}
where $\E(\tilde{\eps}^2\mid \X = \xn) = h(\xn)$. Using 
\eqref{density} again for the first term in \eqref{CVE_problem} obtains formula \eqref{e_LtildeVs0} and \eqref{tilde_eps_var}. 

To see that \eqref{e_objective}, \eqref{e_LtildeVs0}, and \eqref{tilde_eps_var} are well defined and continuous, let $\tilde{g}(\V,\bs_0,\rs)= g(\B^T\bs_0 + \B^T\V\rs)^l f_\X(\bs_0 + \V\rs)$ for $l = 1,2$ or $\tilde{g}(\V,\bs_0,\rs)= h(\B^T\bs_0 + \B^T\V\rs) f_\X(\bs_0 + \V\rs)$ (for \eqref{tilde_eps_var}) which are continuous by assumption. In consequence, the parameter depending integrals 
\eqref{tl} and \eqref{tilde_eps_var} are well defined and continuous if (1) $\tilde{g}(\V,\bs_0,\cdot)$ is integrable for all $\V \in \spc(p,q),\bs_0 \in \text{supp}(f_\X)$, (2) $\tilde{g}(\cdot,\cdot,\rs)$ is continuous for all $\rs$, and (3) there exists an integrable dominating function of $\tilde{g}$ that does not depend on  $\V$ and $\bs_0$ [see \cite[p. 101]{HarroHeuser}]. 

Furthermore, for some compact set $\mathcal{K}$, $t^{(l)}(\V,\bs_0) = \int_{\mathcal{K}} \tilde{g}(\V,\bs_0,\rs) d\rs$, since $\text{supp}(f_\X)$ is compact by (E.2). The function $\tilde{g}(\V,\bs_0,\rs)$ is continuous in all inputs by the continuity of $g$ (E.4) and $f_\X$ by (E.2), and therefore it attains a maximum. In consequence, all three conditions are satisfied so that $t^{(l)}(\V,\bs_0)$ is well defined and continuous. By the same argument \eqref{tilde_eps_var} is well defined and continuous.

Next, $\mu_l(\V,\bs_0) = t^{(l)}(\V,\bs_0)/t^{(0)}(\V,\bs_0)$ is continuous since $t^{(0)}(\V,\bs_0) > 0$ for all interior points $\bs_0 \in \text{supp}(f_\X)$ by the continuity of $f_\X$, convexity of the support and $\Sigmaxbf > 0$. Then, $\tilde{L}(\V,\bs_0,f)$ in  \eqref{e_LtildeVs0} is continuous, which results in \eqref{e_LV1} also being well defined and continuous by virtue of it being a parameter depending integral following the same arguments as above. Moreover, \eqref{CVE_of_transformed_Y} exists as the minimizer of a continuous function over the compact set $\spc(p,q)$. 
\medskip
Then, \eqref{e_LV1} can be writen as
\begin{align} \label{LV_Fa}
    L_\Fa^*(\V,f) = \E_{\bs_0 \sim \X}\left(\mu_2(\V,\bs_0,f) - \mu_1(\V,\bs_0,f)^2 \right) + \E_{\bs_0 \sim \X}\left(\var\left(\tilde{\eps} \mid \X \in \bs_0 + \spn\{\V\} \right) \right)
\end{align}
where $\bs_0 \sim \X$ signifies that $\bs_0$ is distributed as $\X$ and the expectation is with respect to the distribution of $\bs_0$. 

It now suffices to show that the second term on the right hand side of \eqref{LV_Fa} is constant with respect to $\V$.
By the law of total variance,
\begin{align}
    \var(\tilde{\eps}) &= \E\left(\var(\tilde{\eps}\mid \X \in \bs_0 + \spn\{\V\})\right) + \var\left(\E(\tilde{\eps}\mid \X \in \bs_0 + \spn\{\V\} )\right) \notag\\
   &= \E\left(\var(\tilde{\eps}\mid \X \in \bs_0 + \spn\{\V\})\right) \label{variance_constant}
\end{align}
since $\E(\tilde{\eps}\mid \X \in \bs_0 + \spn\{\V\} ) = \E(\underbrace{\E(\tilde{\eps}\mid \X)}_{= 0}\mid \X \in \bs_0 + \spn\{\V\} ) = 0$.  Inserting \eqref{variance_constant} into \eqref{LV_Fa} obtains
\begin{align}
L_\Fa^*(\V,f_t) = \E\left(\mu_2(\V,\X,f_t) - \mu_1(\V,\X,f_t)^2\right) + \var(\tilde{\eps}) \notag \\
= \E_{\bs_0 \sim \X}\left(\var\left(g(\widetilde{\B}^T\X) \mid \X \in \bs_0 + \spn\{\V\}\right) \right)  + \var(\tilde{\eps})\label{CVE_target}
\end{align}
Next we show that \eqref{e_LV1}, or, equivalently \eqref{CVE_target}), attains its minimum at $\V \perp \widetilde{\B}$. Let $\bs_0 \in \text{supp}(f_\X)$ and $\V = (\vb_1,...,\vb_q) \in \real^{p \times q}$,  so that $\vb_u \in \spn\{\B\}$ for some $u \in \{1,...,q\}$. 
Since $\X \in \bs_0 +\spn\{\V\} \Longleftrightarrow \X = \bs_0 + \Pbf_{\V}(\X - \bs_0)$, by the first term in \eqref{CVE_target}
\begin{align}
&\Var\left(g(\widetilde{\B}^T\X) \mid \X \in \bs_0 + \spn\{\V\}\right) = \Var\left(g(\widetilde{\B}^T\X) \mid \X = \bs_0 + \V\V^T(\X-\bs_0)\right) 
\notag \\
&= \Var\left(g(\widetilde{\B}^T\bs_0 + \widetilde{\B}^T\V\V^T(\X-\bs_0))\mid\X = \bs_0 + \V\V^T(\X-\bs_0)\right) \geq 0 \label{1stterm}
\end{align}
If \eqref{1stterm} is positive, i.e. $\widetilde{\B}^T\V\V^T(\X-\bs_0) \neq 0$ with positive probability, then the lower bound is not attained. If it is zero; i.e., for  $\V$ such that  $\V $ and $\widetilde{\B}^T$ are orthogonal, then $L_\Fa^*(\V,f) = \var(\tilde{\eps})$. 
Since $\bs_0$ is arbitrary yet constant, the same inequality holds for \eqref{e_LV1}; that is, \eqref{e_LV1} attains its minimum for $\V$ such that  $\V $ and $\widetilde{\B}^T$ are orthogonal. Since $\spn\{\widetilde{\B}^T\} = \msf$,  \eqref{CVE_of_transformed_Y} follows.
\end{proof}

\begin{proof}[Proof of Theorem~\ref{ECVE_identifies_cs_thm}]
Under assumptions (E.1), (E.2), and (E.3), \eqref{e_objective} is well defined and continuous by arguments analogous to those in the proof of Theorem~\ref{CVE_targets_meansubspace_thm}. Therefore \eqref{enVq}  exists as a minimizer of a continuous function over the compact set $\spc(p,q)$.

To show $\cs = \spn\{\V_q\}^\perp$, let $\tilde{\spc} \neq \cs$ with $\dim(\tilde{\spc}) = \dim(\cs) = k$. Also, let $\Z \in \real^{p \times (p-k)}$ be an orthonormal base of $\tilde{\spc}^\perp$.
Suppose $L_\Fa(\Z) = \min_{V \in \spc(p,p-k)}L_\Fa(\V)$. 
By~\eqref{CVE_of_transformed_Y} and~\eqref{CVE_targets_meansubspace} in  Theorem~\ref{CVE_targets_meansubspace_thm}, $L_\Fa^*(\V,f_t)$, considered as a function from $\real^{p \times (p -k_t)}$, is minimized by an orthonormal base of $\msf^\perp$ with  $p -k_t$ elements, where $k_t = \dim(\msf) \leq k$. By (E.1), $\msf \subseteq \cs= \spn\{\B\}$. As in the proof of Theorem~\ref{CVE_targets_meansubspace_thm}, we obtain that $L_\Fa^*(\V,f_t)$, as a function from $\real^{p \times (p-k)}$, is minimized by an orthonormal base $\Ub \in \real^{p \times (p-k)}$ of $\spn\{\B\}^\perp$.

Since $\tilde{\spc}= \spn\{\Z\} \neq \spn\{\Ub\} = \cs$, we can rearrange the bases $\Ub = (\Ub_1,\Ub_2)$ and $\Z = (\Z_1,\Z_2)$ such that $\spn\{\Ub_1\} = \spn\{\Z_1\}$ and $\spn\{\Ub_2\} \neq \spn\{\Z_2\}$. Since $\Fa$ characterises $\cs$, the set $A = \{t \in \Omega_T: \spn\{\Ub_2\} \subseteq \msf \}$ is non empty and by (E.3) $A$ is not a null set with respect to the probability measure $F_T$.

Thus, 
\begin{gather*}
\min_{V \in \spc(p,p-k)}L_\Fa(\V) = L_\Fa(\Z) = \E_{t\sim F_T}\left(L_\Fa^*(\Z,f_t)\right) \\
= \int_A \underbrace{L_\Fa^*(\Z,f_t)}_{> L_\Fa^*(\Ub,f_t)} dF_T(t) +\int_{A^c} \underbrace{L_\Fa^*(\Z,f_t)}_{= L_\Fa^*(\Ub,f_t)}dF_T(t) >\E_{t\sim F_T}\left(L_\Fa^*(\Ub,f_t)\right),
\end{gather*}
which contradicts our assumption that $L_\Fa(\Z) = \min_{V \in \spc(p,p-k)}L_\Fa(\V)$. 
\end{proof}



\smallskip
\noindent
Next we introduce notation and auxiliary lemmas for the proof of Theorem~\ref{uniform_convergence_ecve}. We suppose all assumptions of Theorem~\ref{uniform_convergence_ecve} hold. We generically use the letter ``C'' to denote constants. 

Suppose  $f$ is an arbitrary element of  $\Fa$ and let
\begin{align}\label{Ytilde_decomposition}
   \tilde{Y}_i = f(Y_i) = g(\widetilde{\B}^T\X_i) + \tilde{\eps_i} 
\end{align}
with $\spn\{\widetilde{\B}\} = \spc_{\E(\tilde{Y}\mid \X)} = \spc_{\E(f(Y)\mid \X)}$. Condition (E.4) yields that $g$ is twice continuously differentiable, and $\E(|\tilde{Y}|^8) < \infty$. Since $f$ is fixed, we suppress it in $t^{(l)}(\V,\bs_0,f)$ and $\tilde{h}(\V,\bs_0,f)$, so that  
\begin{align}\label{tn}
t^{(l)}_n(\V,\bs_0,f) = t^{(l)}_n(\V,\bs_0) =\frac{1}{nh_n^{(p-q)/2}}\sum_{i=1}^n K\left(\frac{d_i(\V,\bs_0)}{h_n}\right)\tilde{Y}^l_i,
\end{align}
which is  the sample version of \eqref{tl} for $l=0,1,2$. Eqn. \eqref{e_ybar} can be expressed  as
\begin{align}\label{yl}
\bar{y}_l(\V,\bs_0)
&= \frac{ t^{(l)}_n(\V,\bs_0)}{t^{(0)}_n(\V,\bs_0)},
\end{align}



\begin{lemma}\label{aux_lemma2}
Assume  (E.2) and (K.1) hold. For a continuous function $g$, we let $Z_n(\V,\bs_0) =  \left(\sum_i g(\X_i)^l K(d_i(\V,\bs_0)/h_n)\right) /(n h_n^{(p-q)/2})$. Then, 
\begin{align*}
    \E\left(Z_n(\V,\bs_0)\right) 
    &= \int_{\text{supp}(f_\X)\cap\real^{p-q}}K(\|\rs_2\|^2)\int_{\text{supp}(f_\X)\cap\real^q} \tilde{g}(\rs_1,h_n^{1/2}\rs_2)d\rs_1 d\rs_2
\end{align*}
\end{lemma}
where $\tilde{g}(\rs_1,\rs_2) = g(\bs_0 + \V\rs_1 +\Ub\rs_2)^lf_\X(\bs_0 + \V\rs_1 + \Ub\rs_2)$, $\xn= \bs_0 + \V \rs_1 + \Ub \rs_2$  in \eqref{ortho_decomp}.
\begin{proof}[Proof of Lemma~\ref{aux_lemma2}]
By \eqref{ortho_decomp}, $\|\Pbf_{\Ub} (\xn-\bs_0)\|^2 = \|\Ub \rs_2\|^2 = \|\rs_2\|^2$. Further 
\begin{gather*} 
\E\left(Z_n(\V,\bs_0)\right) = \frac{1}{h_n^{(p-q)/2}} \int_{\text{supp}(f_\X)} g( \xn)^l K(\|\Pbf_\Ub (\xn-\bs_0)/ h^{1/2}_n \|^2) f_\X(\xn)d\xn \\ 
= \frac{1}{h_n^{(p-q)/2}} \int_{\text{supp}(f_\X)\cap\real^{p-q}}\int_{\text{supp}(f_\X)\cap\real^q} g(\bs_0 + \V \rs_1 +  \Ub \rs_2)^lK(\|\rs_2/ h^{1/2}_n\|^2) \times \\ f_\X(\bs_0 + \V \rs_1 + \Ub \rs_2)d\rs_1 d\rs_2\\ \notag
=  \int_{\text{supp}(f_\X)\cap\real^{p-q}}K(\|\rs_2\|^2)\int_{\text{supp}(f_\X)\cap\real^q} g(\bs_0 + \V \rs_1 + h_n^{1/2} \Ub \rs_2)^l 
f_\X(\bs_0 + \V \rs_1 + h_n^{1/2}\Ub\rs_2)d\rs_1 d\rs_2
\end{gather*}
where the substitution $\tilde{\rs}_2 = \rs_2/h_n^{1/2}$, $d\rs_2 = h_n^{(p-q)/2} d\tilde{\rs}_2$ was used to obtain the last equality.
\end{proof}

\begin{lemma}\label{aux_lemma3}
Assume (E.1), (E.2), (E.3), (E.4), (H.1) and (K.1) hold. Then, there exists a constant $C > 0$, such that  
\begin{equation*}
   \var\left(n h_n^{(p-q)/2} t^{(l)}_n(\V,\bs_0,f)\right) \leq n h_n^{(p-q)/2} C
\end{equation*}
\end{lemma}
for $n > n^\star$ and $t^{(l)}_n(\V,\bs_0)$, $l = 0,1,2$, in \eqref{tn}.
\begin{proof}[Proof of Lemma~\ref{aux_lemma3}]
Since a continuous function attains a finite maximum over a compact set, $ \sup_{\xn \in \text{supp}(f_\X)}|g(\widetilde{\B}^T\xn)| < \infty.$ Therefore,
\begin{align*}
    |\tilde{Y}_i| \leq |g(\widetilde{\B}^T\X_i)| +|\tilde{\eps_i}| \leq \sup_{\xn \in \text{supp}(f_\X)}|g(\widetilde{\B}^T\xn)| +|\tilde{\eps_i}| = C +|\tilde{\eps_i}|
\end{align*}
and $| \tilde{Y}_i|^{2l} \leq  \sum_{u=0}^{2l} \binom{2l}{u} C^u |\tilde{\eps_i}|^{2l - u}$. 
Since $(\tilde{Y}_i,\X_i)$ are i.i.d.,
\begin{gather}
    \var\left(n h_n^{(p-q)/2} t^{(l)}_n(\V,\bs_0,f)\right) =  n \var\left(\tilde{Y}^l K\left(\frac{d_i(\V,\bs_0)}{h_n}\right)\right) \leq
    n \E\left(\tilde{Y}^{2l} K^2\left(\frac{d_i(\V,\bs_0)}{h_n}\right)\right) \notag \\
    = n \E\left(|\tilde{Y}|^{2l} K^2\left(\frac{d_i(\V,\bs_0)}{h_n}\right)\right) \leq n \sum_{u=0}^{2l} \binom{2l}{u} C^u \E\left( |\tilde{\eps_i}|^{2l - u} K^2\left(\frac{d_i(\V,\bs_0)}{h_n}\right)\right) \notag \\
    = n \sum_{u=0}^{2l} \binom{2l}{u} C^u \E\left( \E(|\tilde{\eps_i}|^{2l - u}\mid \X_i) K^2\left(\frac{d_i(\V,\bs_0)}{h_n}\right)\right) \label{bound_tnf_var}
\end{gather}
for $l = 0,1,2$.
Let $\E(|\tilde{\eps_i}|^{2l - u} \mid \X_i) = g_{2l-u}(\X_i)$ for a continuous (by assumption) function $g_{2l-u}(\cdot)$ with finite moments for $l=0,1,2$ by the compactness of $\text{supp}(f_\X)$. Using Lemma~\ref{aux_lemma2} with \[Z_n(\V,\bs_0) = \frac{1}{nh_n^{(p-q)/2}} \sum_i g_{2l-u}(\X_i) K^2\left(d_i(\V,\bs_0)/h_n\right),\] where $K^2(\cdot)$ fulfills (K.1), we calculate 
\begin{gather}
    \E\left( \E(|\tilde{\eps_i}|^{2l - u}\mid \X_i) K^2\left(\frac{d_i(\V,\bs_0)}{h_n}\right)\right) = h_n^{(p-q)/2} \E(Z_n(\V,\bs_0)) \notag \\
    = h_n^{(p-q)/2} \int_{\text{supp}(f_\X)\cap\real^{p-q}}K^2(\|\rs_2\|^2)\times \notag \\ 
    \quad \int_{\text{supp}(f_\X)\cap\real^q} g_{2l-u}(\bs_0 + \V \rs_1 + h_n^{1/2} \Ub \rs_2)
f_\X(\bs_0 + \V \rs_1 + h_n^{1/2}\Ub\rs_2)d\rs_1 d\rs_2 \label{ugly_integral}\\
\leq h_n^{(p-q)/2} C \notag
\end{gather}
since all integrands in \eqref{ugly_integral} are continuous and over compact sets by (E.2) and the continuity of $g_{2l-u}(\cdot)$ and $K(\cdot)$, so that  the integral can be upper bounded by a finite constant $C$. Inserting $\eqref{ugly_integral}$ into \eqref{bound_tnf_var} yields
\begin{align}\label{upper_bound_var_tnf}
    \var\left(n h_n^{(p-q)/2} t^{(l)}_n(\V,\bs_0,f)\right) \leq n h_n^{(p-q)/2} \underbrace{\sum_{u=0}^{2l} \binom{2l}{u} C^u C}_{ = C} = n h_n^{(p-q)/2}C
\end{align}
\end{proof}

In Lemma \ref{d_inequality} we show that $d_i(\V,\bs_0)$ in \eqref{distance} is Lipschitz in its inputs under assumption (E.2). 

\begin{lemma}\label{d_inequality}
Under assumption (E.2) there exists a constant $0 < C_2 < \infty$ such that for all $\delta >0$ and  $\V, \V_j \in \spc(p,q)$ with $\|\Pbf_\V - \Pbf_{\V_j}\| < \delta$ and for all $\bs_0, \bs_j \in \text{supp}(f_\X) \subset \real^p$ with $\|\bs_0 - \bs_j\| < \delta$,
\begin{equation*}
    |d_i(\V,\bs_0) - d_i(\V_j,\bs_j)| \leq C_2 \delta
\end{equation*}
for $d_i(\V,\bs_0)$ given by \eqref{distance}
\end{lemma}
\begin{proof}[Proof of Lemma~\ref{d_inequality}]
\begin{gather}
|d_i(\V,\bs_0) - d_i(\V_j,\bs_j)| \leq \left|\|\X_i - \bs_0\|^2 - \|\X_i - \bs_j\|^2\right| + \notag\\
\left|\langle \X_i - \bs_0,\Pbf_{\V}(\X_i - \bs_0)\rangle - \langle \X_i - \bs_j,\Pbf_{\V_j}(\X_i - \bs_j)\rangle\right| = I_1 + I_2\label{di-dj2}
\end{gather}
where $\langle\cdot,\cdot \rangle$ is the scalar product in $\real^p$.  We bound  the first term on the right hand side of \eqref{di-dj2} as follows using $\|\X_i \| \leq \sup_{z \in \text{supp}(f_\X)} \|z \| = C_1 < \infty$ with probability 1 by (E.2).
\begin{align*}
I_1 &= \left|\|\X_i - \bs_0\|^2 - \|\X_i - \bs_j\|^2\right| \leq 
2\left|\langle \X_i,\bs_0 -\bs_j \rangle\right| + \left|\|\bs_0\|^2 - \|\bs_j\|^2\right| \\
& \leq 2\|\X_i\|\|\bs_0 -\bs_j\| + 2C_1\|\bs_0 - \bs_j\| 
\leq 2C_1 \delta + 2C_1\delta = 4 C_1\delta
\end{align*}
by Cauchy-Schwartz and the reverse triangular inequality for which  $\left|\|\bs_0\|^2 - \|\bs_j\|^2\right| = \left|\|\bs_0\| - \|\bs_j\|\right|(\|\bs_0\| + \|\bs_j\|) \leq \|\bs_0 - \bs_j\|2C_1$. 
The second term in \eqref{di-dj2} satisfies
\begin{gather*}
I_2 \leq \left|\langle \X_i,(\Pbf_{\V}-\Pbf_{\V_j})\X_i\rangle\right| + 2\left|\langle \X_i,\Pbf_{\V}\bs_0 -\Pbf_{\V_j}\bs_j\rangle\right| + \left|\langle \bs_0,\Pbf_{\V}\bs_0\rangle - \langle \bs_j,\Pbf_{\V_j}\bs_j\rangle\right| \\
\leq \|\X_i\|^2\|\Pbf_{\V} - \Pbf_{\V_j}\| + 2\|\X_i\| \left\|\Pbf_{\V}(\bs_0 - \bs_j) + (\Pbf_{\V}-\Pbf_{\V_j})\bs_j\right\| + \left|\langle \bs_0 - \bs_j,\Pbf_{\V} \bs_0 \rangle\right|+\\
\left|\langle \bs_j,\Pbf_{\V} \bs_0 - \Pbf_{\V_j}\bs_j \rangle\right| \leq C_1^2 \delta + 2C_1(\delta + C_1 \delta) + C_1\delta +C_1 (\delta + C_1 \delta) =4C_1\delta + 4C_1^2 \delta
\end{gather*}
Collecting all constants into $C_2$ (i.e. $C_2 = 8C_1 + 4C_1^2$) yields the result.
\end{proof}
To show Theorems~\ref{uniform_convergence_ecve} and \ref{thm_variance}, we use the \textbf{Bernstein inequality} \cite{Bernstein}. 
Let $\{Z_i, i=1,2,\ldots \}$, be an independent sequence of bounded random variables with $|Z_i| \leq b$. Let $S_n = \sum_{i=1}^n Z_i$, $E_n = \E(S_n)$ and $V_n = \var(S_n)$. Then,
\begin{equation}\label{Bernstein}
P(|S_n - E_n|>t) < 2 \exp{\left(-\frac{t^2/2}{V_n + b t/3} \right)}
\end{equation}
Assumption (K.2) yields
\begin{equation} \label{kernel}
    |K(u) - K(u')| \leq K^*(u') \delta
\end{equation}
for all $u, u'$ with $|u-u'| < \delta \leq L_2$ and $K^*(\cdot)$ is a bounded and integrable kernel function [see \cite{Hansen2008}]. Specifically, if condition (1) of (K.2) holds, then $K^*(u) = L_1 1_{\{|u| \leq 2L_2\}}$. If condition (2) holds, then  $K^*(u) = L_1 1_{\{|u| \leq 2L_2\}} + 1_{\{|u| > 2L_2\}}|u-L_2|^{-\nu}$.

Let $A = \spc(p,q) \times \text{supp}(f_\X)$. In Lemma~\ref{thm_variance} and \ref{thm_bias} we show that \eqref{tn} converges uniformly in probability to \eqref{tl} by showing that the variance and bias terms vanish uniformly in probability, respectively.

\begin{lemma}\label{thm_variance}
Under the assumptions of Theorem~\ref{uniform_convergence_ecve},  
\begin{equation}
    \sup_{\V \times \bs_0 \in A} \left|t^{(l)}_n(\V,\bs_0) - \E\left(t_n^{(l)}(\V,\bs_0)\right)\right| = O_{P}(a_n), \; l=0,1,2
\end{equation}
\end{lemma}

\begin{proof}[Proof of Lemma~\ref{thm_variance}]
The proof 
proceeds in 3 steps: (i) truncation, (ii)  discretization by covering $A= \spc(p,q) \times \text{supp}(f_\X)$, and (iii) application of Bernstein's inequality \eqref{Bernstein}. If the function $f$ in \eqref{Ytilde_decomposition} is bounded, the truncation step and the assumption $a_n/h_n^{(p-q)/2} = O(1)$ are not needed.

(i) We let $\tau_n = a_n^{-1}$ and  truncate $\tilde{Y}_i^l$ by $\tau_n$ as follows. We let  \begin{align}\label{tl.trc}
    t^{(l)}_{n,\text{trc}}(\V,\bs_0) &= (1/nh_n^{(p-q)/2})\sum_i K(\|\Pbf_\Ub (\X_i-\bs_0) \|^2/ h_n)\tilde{Y}_i^l 1_{\{|\tilde{Y}_i|^l \leq \tau_n\}}
\end{align} be the truncated version of \eqref{tn} and $\tilde{R}^{(l)}_n = (1/nh_n^{(p-q)/2})\sum_i |\tilde{Y}_i|^l 1_{\{|\tilde{Y}_i|^l > \tau_n\}}$ be the remainder of \eqref{tn}. Therefore $R^{(l)}_n(\V,\bs_0) = t^{(l)}_n(\V,\bs_0) - t^{(l)}_{n,\text{trc}}(\V,\bs_0) \leq M_1 \tilde{R}^{(l)}_n $ due to (K.1)
and
\begin{align}
\sup_{\V \times \bs_0 \in A} \left|t^{(l)}_n(\V,\bs_0) - \E\left(t_n^{(l)}(\V,\bs_0)\right)\right| &\leq M_1(\tilde{R}^{(l)}_n +\E\tilde{R}^{(l)}_n) \notag \\
& \qquad + \sup_{\V \times \bs_0 \in A}\left|t^{(l)}_{n,\text{trc}}(\V,\bs_0) - \E \left( t^{(l)}_{n,\text{trc}}(\V,\bs_0)\right)\right|\label{truncation}
\end{align}
By Cauchy-Schwartz and the Markov inequality, $\Pb(|Z| > t) = \Pb(Z^4 > t^4) \leq \E(Z^4)/t^4$, 
we obtain 
\begin{align}
  \E\tilde{R}^{(l)}_n &= \frac{1}{h_n^{(p-q)/2}}  \E \left(|\tilde{Y}_i|^{l} 1_{\{|\tilde{Y}_i|^l > \tau_n\}}\right) \leq 
  \frac{1}{h_n^{(p-q)/2}}\sqrt{\E(|\tilde{Y}_i|^{2l})} \sqrt{\Pb(|\tilde{Y}_i|^l > \tau_n)} \notag \\
  &\leq \frac{1}{h_n^{(p-q)/2}} \sqrt{\E(|\tilde{Y}_i|^{2l})} \left(\frac{\E(|\tilde{Y}_i|^{4l})}{a_n^{-4}}\right)^{1/2} 
  = o(a_n), \label{Rtilde}
\end{align}
where the last equality uses the assumption $a_n/h_n^{(p-q)/2} = O(1)$ and the expectations are finite due to (E.4) for $l=0,1,2$. No truncation is needed for $l=0$ or if $\tilde{Y}_i = f(Y_i) \leq \sup_{f \in \Fa} |f(Y_i)| < C < \infty$.

Therefore, the first two terms of the right hand side of \eqref{truncation} converge to 0 with rate $a_n$ by \eqref{Rtilde} and Markov's inequality. From this point on, $\tilde{Y}_i$ will denote the truncated version $\tilde{Y}_i 1_{\{|\tilde{Y}_i| \leq \tau_n\}}$ and we do not distinguish the truncated from the untruncated $t_n(\V,\bs_0)$ since this truncation results in an error of magnitude $a_n$. 

(ii) For the discretization step we cover the compact set $A = \spc(p,q) \times \text{supp}(f_\X)$ by finitely many balls, which is possible by (E.2) 
 and the compactness of $\spc(p,q)$. 
Let $\delta_n = a_n h_n$ and $A_j = \{\V: \|\Pbf_\V - \Pbf_{\V_j}\| \leq \delta_n\} \times \{\bs :\|\bs - \bs_j\| \leq \delta_n\}$ be a cover of $A$ with ball centers $\V_j \times \bs_j$. Then, $A \subset \bigcup_{j=1}^{N} A_j$ and the number of balls can be bounded by $N \leq C \, \delta_n^{-d}\delta_n^{-p}$ for some constant $C \in (0, \infty)$, where $d = \text{dim}(\spc(p,q)) = pq - q(q+1)/2$. 
Let $\V \times \bs_0 \in A_j$.  Then by Lemma~\ref{d_inequality} there exists  $0 < C_2 < \infty$, such that
\begin{align}\label{inequality1}
|d_i(\V,\bs_0) - d_i(\V_j,\bs_j)| \leq C_2 \delta_n
\end{align}
for $d_i$ in \eqref{distance}. Under (K.2), which implies \eqref{kernel},  inequality~\eqref{inequality1} yields
\begin{equation}\label{Ki-Kj}
\left|K\left(\frac{d_i(\V,\bs_0)}{h_n}\right) - K\left(\frac{d_i(\V_j,\bs_j)}{h_n}\right)\right| \leq K^*\left(\frac{d_i(\V_j,\bs_j)}{h_n}\right) C_2 a_n 
\end{equation}
for $\V \times \bs_0 \in A_j$ and $K^*(\cdot)$ an integrable and bounded function. 

Define $r^{(l)}_n(\V_j,\bs_j) = (1/nh_n^{(p-q)/2}) \sum_{i=1}^n K^*(d_i(\V_j,\bs_j)/h_n)|\tilde{Y}_i|^l$. For notational convenience we next drop the dependence on $l$ and $j$ and observe that \eqref{Ki-Kj} yields 
\begin{equation}\label{t_diff}
|t^{(l)}_n(\V,\bs_0) - t^{(l)}_n(\V_j,\bs_j)| \leq C_2 a_n r^{(l)}_n(\V_j,\bs_j)
\end{equation}
Since $K^*$ fulfills (K.1) except for continuity, an analogous argument as in the proof of Lemma~\ref{aux_lemma2} yields that $\E\left(r^{(l)}_n(\V_j,\bs_j)\right) < \infty$. By subtracting and adding $t^{(l)}_n(\V_j,\bs_j)$, $\E(t^{(l)}_n(\V_j,\bs_j))$, the triangular inequality, \eqref{t_diff} and integrability of $r_n^l$, we obtain 
\begin{gather}
\left|t^{(l)}_n(\V,\bs_0) - \E\left(t_n^{(l)}(\V,\bs_0)\right)\right| \leq 
\left|t^{(l)}_n(\V,\bs_0) - t^{(l)}_n(\V_j,\bs_j)\right| +  \left|\E\left(t^{(l)}_n(\V_j,\bs_j) - t_n^{(l)}(\V,\bs_0)\right)\right| \notag\\+ \left|t^{(l)}_n(\V_j,\bs_j) - \E\left(t_n^{(l)}(\V_j,\bs_j)\right)\right|  
\leq C_2 a_n \left(|r_n| +   |\E\left(r_n\right)| \right) + \left|t^{(l)}_n(\V_j,\bs_j) - \E\left(t_n^{(l)}(\V_j,\bs_j)\right)\right|  \notag\\
\leq C_2 a_n(|r_n -\E(r_n)| + 2|\E(r_n)|) + \left|t^{(l)}_n(\V_j,\bs_j) - \E\left(t_n^{(l)}(\V_j,\bs_j)\right)\right| \notag\\ 
\leq 2C_3a_n + |r_n -\E(r_n)| + \left|t^{(l)}_n(\V_j,\bs_j) - \E\left(t_n^{(l)}(\V_j,\bs_j)\right)\right| \label{inequality2}
\end{gather}
for any constant  $C_3 > C_2 \E(r^{(l)}_n(\V_j,\bs_j))$ and  $n$ such that $C_2 a_n \leq 1$, since $a_n^2 = o(1)$, which in turn yields that there exists $0 < C_3 < \infty$ such that \eqref{inequality2} holds. 

Since $\sup_{x \in A} f(x) = \max_{1\leq j\leq N}\sup_{x \in A_j}f(x) \leq \sum_{j=1}^N \sup_{x \in A_j}f(x)$ for any cover of $A$ and continuous function $f$, 
\begin{gather}
\Pb(\sup_{\V \times \bs_0 \in A} |t^{(l)}_n(\V,\bs_0) - \E\left(t_n^{(l)}(\V,\bs_0)\right)| > 3C_3a_n) \notag\\ 
\leq \sum_{j=1}^N \Pb(\sup_{\V \times \bs_0 \in A_j} |t^{(l)}_n(\V,\bs_0) - \E\left(t_n^{(l)}(\V,\bs_0)\right)| > 3C_3a_n) \notag \\ \leq N \max_{1 \leq j \leq N} \Pb(\sup_{\V \times \bs_0 \in A_j} |t^{(l)}_n(\V,\bs_0) - \E\left(t_n^{(l)}(\V,\bs_0)\right)| > 3C_3a_n) \label{Prob1} \\ 
\leq N \left(\max_{1 \leq j \leq N}\Pb(|t^{(l)}_n(\V_j,\bs_j) - \E\left(t_n^{(l)}(\V_j,\bs_j)\right)| > C_3 a_n) +  \max_{1 \leq j \leq N} \Pb(|r_n -\E(r_n)| > C_3a_n)\right)  \leq \notag\\ 
C\, \delta^{-(d+p)} \left( \max_{1 \leq j \leq N}\Pb(|t^{(l)}_n(\V_j,\bs_j) - \E\left(t_n^{(l)}(\V_j,\bs_j)\right)| > C_3 a_n) + \max_{1 \leq j \leq N} \Pb(|r_n -\E(r_n)| > C_3a_n)\right) \notag
\end{gather}
by the subadditivity of  probability for the first inequality and \eqref{inequality2} for the third inequality above, where the last inequality is due to $N \leq C\, \delta_n^{-d}\delta_n^{-p}$ for a cover of $A$.

Finally, we bound the first and second term in the last line of \eqref{Prob1} by the Bernstein inequality~\eqref{Bernstein}. For the first term in the last line of \eqref{Prob1}, let $Z_i = Y^l_i K(d_i(\V_j,\bs_j)/h_n)$ and $S_n = \sum_i Z_i = nh_n^{(p-q)/2} t^{(l)}_n(\V_j,\bs_j)$. Then, $Z_i$ are independent with $|Z_i| \leq b =  M_1\tau_n =M_1/a_n$ by (K.1) and the truncation step (i). For $V_n = \var(S_n)$, Lemma~\ref{aux_lemma3} yields $nh_n^{(p-q)/2}C  \geq V_n$ with $C>0$
, and set $t = C_3a_n n h_n^{(p-q)/2}$. The  Bernstein inequality~\eqref{Bernstein} yields
\begin{gather*} \label{first_term}
\Pb\left(\left|t^{(l)}_n(\V_j,\bs_j) - \E\left(t_n^{(l)}(\V_j,\bs_j)\right)\right| > C_3 a_n\right)  < 
2 \exp{\left(\frac{-t^2/2}{V_n + b t/3}\right)} \leq \\
2 \exp{\left(-\frac{(1/2)C_3^2a^2_n n^2 h_n^{(p-q)}}{nh_n^{(p-q)/2}C  + (1/3) M_1\tau_n C_3 a_n n h_n^{(p-q)/2})} \right)} \leq 
 2 \exp{\left(-\frac{(1/2)C_3\log(n)}{C/C_3 + M_1/3 } \right)} = 
2 n^{-\gamma(C_3)}
\end{gather*}
where  $a_n^2 = \log(n)/(n h_n^{(p-q)/2})$ and  $\gamma(C_3) = C_3\left(2(C/C_3 + M_1 /3)\right)^{-1} $ that is  an increasing function that can be made arbitrarily large by increasing  $C_3$.

For the second term in the last line of \eqref{Prob1}, set $Z_i = Y^l_i K^*(d_i(\V_j,\bs_j)/h_n)$ 
in~\eqref{Bernstein} and proceed similarly to obtain
\begin{gather*} \label{second_term}
\Pb\left(\left|r^{(l)}_n(\V_j,\bs_j) - \E\left(r_n^{(l)}(\V_j,\bs_j)\right)\right| > C_3 a_n\right) < 2 n^{- \frac{(1/2)C_3}{C/C_3 + (1/3) M_2}} = 2 n^{-\gamma(C_3)}
\end{gather*}
By (H.1), $h_n^{(p-q)/4} \leq 1$ for $n$ large and (H.2) implies $1/(nh_n^{(p-q)/2}) \leq 1$ for $n$ large, therefore $h_n^{-1} \leq n^{2/(p-q)} \leq n^2$ since $p-q \geq 1$. Then, $\delta_n^{-1} = (a_n h_n)^{-1} \leq n^{1/2}h_n^{-1} h_n^{(p-q)/4} \leq n^{5/2}$.  
Therefore, \eqref{Prob1} is smaller than $4\,C\, \delta_n^{-(d+p)}n^{-\gamma(C_3)} \leq 4C n^{5(d+p)/2 - \gamma(C_3)}$. For $C_3$ large enough, we have $5(d+p)/2 - \gamma(C_3) < 0$ and $n^{5(d+p)/2 - \gamma(C_3)} \to 0$.  
This completes the proof. 
\end{proof}

If we assume $|\tilde{Y}_i| < M_2 < \infty$ almost surely, the requirement $a_n/h_n^{(p-q)/2} = O(1)$ for the bandwidth can be dropped and the truncation step of the proof of Lemma~\ref{thm_variance} is no longer necessary.

\begin{lemma}\label{thm_bias}
Under 
(E.1), (E.2), (E.3), (E.4), (H.1), (K.1), and 
$\int_{\real^{p-q}}K(\|\rs_2\|^2)d\rs_2 = 1$,
\begin{equation}
    \sup_{\V \times\bs_0 \in A} \left| t^{(l)}(\V,\bs_0)+1_{\{l=2\}}\tilde{h}(\V,\bs_0) 
    - \E\left(t_n^{(l)}(\V,\bs_0)\right)\right| =O(h_n), \quad l=0,1,2
\end{equation}
where $t^{(l)}(\V,\bs_0)$ and $\tilde{h}(\V,\bs_0)$ are defined in Theorem~\ref{CVE_targets_meansubspace_thm}.
\end{lemma}
\begin{proof}[Proof of Lemma~\ref{thm_bias}]
Let $\tilde{g}(\rs_1,\rs_2)= g(\widetilde{\B}^T\bs_0 + \widetilde{\B}^T\V\rs_1 +\widetilde{\B}^T\Ub\rs_2)^l f_\X(\bs_0 + \V\rs_1 + \Ub\rs_2)$, where $\rs_1,\rs_2$ satisfy the  orthogonal decomposition~\eqref{ortho_decomp}. 
\begin{align*} \E\left(t_n^{(0)}(\V,\bs_0)\right) &= \E\left(K(d_i(\V,\bs_0)/h_n)\right)/h_n^{(p-q)/2} \\ \E(t_n^{(1)}(\V,\bs_0)) &= \E\left(K(d_i(\V,\bs_0)/h_n)g(\widetilde{\B}^T\X_i)\right)/h_n^{(p-q)/2} \\
&\quad + \E\left(K(d_i(\V,\bs_0)/h_n)\underbrace{\E(\tilde{\eps}_i\mid\X)}_{=0}\right)/h_n^{(p-q)/2}\\
\E(t_n^{(2)}(\V,\bs_0)) &= \E\left(K(d_i(\V,\bs_0)/h_n)g(\widetilde{\B}^T\X_i)^2\right)/h_n^{(p-q)/2} \\
&\quad + 2\E\left(K(d_i(\V,\bs_0)/h_n)\underbrace{\E(\tilde{\eps}_i\mid\X)}_{=0}\right)/h_n^{(p-q)/2} \\
&\qquad+ \E\left(K(d_i(\V,\bs_0)/h_n)\underbrace{\E(\tilde{\eps}^2_i\mid\X)}_{= h(\X_i)}\right)/h_n^{(p-q)/2}
\end{align*}
Then
\begin{align}\label{bias1}
\E\left(t_n^{(l)}(\V,\bs_0)\right) = \int_{\real^{p-q}}K(\|\rs_2\|^2)\int_{\real^p} \tilde{g}(\rs_1,{h_n}^{1/2}\rs_2) d\rs_1 d\rs_2
\end{align}
holds by Lemma~\ref{aux_lemma2} for $l = 0,1$. For $l = 2$, $\tilde{Y}_i^2 = g_i^2 + 2g_i \epsilon_i + \epsilon_i^2$ with $g_i = g(\widetilde{\B}^T\X_i)$ and can be handled as in the case of $l = 0,1$. 
Plugging in \eqref{bias1} the second order Taylor expansion for some $\xi$ in the neighborhood of 0, $\tilde{g}(\rs_1,{h_n}^{1/2}\rs_2) = \tilde{g}(\rs_1,0) + {h_n}^{1/2} \nabla_{\rs_2}\tilde{g}(\rs_1,0)^T\rs_2 + h_n\rs_2^T \nabla^2_{\rs_2}\tilde{g}(\rs_1,\xi) \rs_2$
, yields
\begin{gather*}
\E\left(t_n^{(l)}(\V,\bs_0)\right) = \int_{\real^q}\tilde{g}(\rs_1,0) d\rs_1 + \sqrt{h_n} \left(\int_{\real^q}\nabla_{\rs_2}\tilde{g}(\rs_1,0)d\rs_1\right)^T\int_{\real^{p-q}}K(\|\rs_2\|^2)\rs_2 d\rs_2 + \\
 h_n \frac{1}{2}\int_{\real^{p-q}}K(\|\rs_2\|^2)\int_{\real^p}\rs_2^T \nabla^2_{\rs_2}\tilde{g}(\rs_1,\xi) \rs_2 d\rs_1d\rs_2 =
 t^{(l)}(\V,\bs_0) + h_n \frac{1}{2}R(\V,\bs_0)
\end{gather*}
since $\int_{\real^q}\tilde{g}(\rs_1,0) d\rs_1 = t^{(l)}(\V,\bs_0)$ and $\int_{\real^{p-q}}K(\|\rs_2\|^2)\rs_2 d\rs_2 = 0 \in \real^{p-q}$ due to $K(\|\cdot\|^2)$ being even. Let $R(\V,\bs_0) = \int_{\real^{p-q}}K(\|\rs_2\|^2)\int_{\real^p}\rs_2^T \nabla^2_{\rs_2}\tilde{g}(\rs_1,\xi) \rs_2 d\rs_1d\rs_2$. By (E.4) and (E.2), 
$|\rs_2^T \nabla^2_{\rs_2}\tilde{g}(\rs_1,\xi) \rs_2| \leq C \|\rs_2\|^2$ for $C = \sup_{\xn,\y} \| \nabla^2_{\rs_2}\tilde{g}(\xn,\y)\| < \infty$, since a continuous function over a compact set is bounded. Then, $R(\V,\bs_0) \leq C C_4 \int_{\real^{p-q}}K(\|\rs_2\|^2)\|\rs_2\|^2d\rs_2 < \infty$ for some $C_4 > 0$, since the integral over $\rs_1$ is over a compact set by (E.2). 
\end{proof}

Lemma~\ref{t_uniform} follows directly from Lemmas~\ref{thm_variance} and \ref{thm_bias} and the triangle inequality.

\begin{lemma}\label{t_uniform}
Suppose  
(E.1), (E.2), (E.3), (E.4), (K.1), (K.2), (H.1) hold. If $a_n^2 = \log(n)/nh_n^{(p-q)/2} = o(1)$, and $a_n/h_n^{(p-q)/2} = O(1)$, then for $l=0,1,2$
\begin{equation*}
\sup_{\V \times \bs_0 \in A} \left|t^{(l)}(\V,\bs_0)+1_{\{l=2\}}\tilde{h}(\V,\bs_0) - t_n^{(l)}(\V,\bs_0)\right| = O_P(a_n + h_n) 
\end{equation*}
\end{lemma}


\begin{thm}\label{thm_Ltilde_uniform}
Suppose 
(E.1), (E.2), (E.3), (E.4), (K.1), (K.2), (H.1) hold. Let $a_n^2 = \log(n)/nh_n^{(p-q)/2} = o(1)$,  $a_n/h_n^{(p-q)/2} = O(1)$, then
\begin{equation*}
\sup_{\V \times \bs_0 \in A}\left|\bar{y}_l(\V,\bs_0) - \mu_l(\V,\bs_0)-1_{\{l=2\}}\tilde{h}(\V,\bs_0)\right| = o_P(1)
, \quad l=0,1,2
\end{equation*}
and
\begin{equation}\label{Ltilde_uniform}
\sup_{\V \times \bs_0 \in A}\left|\tilde{L}_{n,\Fa}(\V,\bs_0) - \tilde{L}_\Fa(\V,\bs_0)\right| = o_P(1) 
\end{equation}
where $\bar{y}_l(\V,\bs_0)$, $\mu_l(\V,\bs_0)$, $\tilde{L}_{n,\Fa}(\V,\bs_0)$ and $\tilde{L}_\Fa(\V,\bs_0)$ are defined in \eqref{e_ybar}, \eqref{mu_l},  \eqref{e_Ltilde} and \eqref{e_LtildeVs0}, respectively. 
\end{thm}
\smallskip
\begin{proof}[Proof of Theorem~\ref{thm_Ltilde_uniform}]
Let $\delta_n = \inf_{\V \times \bs_0 \in A_n}t^{(0)}(\V,\bs_0)$, where $t^{(0)}(\V,\bs_0)$ is defined in \eqref{tl}, and $A_n = \spc(p,q) \times \{\xn \in \text{supp}(f_\X): |\xn - \partial\text{supp}(f_\X)| \geq b_n\}$, where $\partial C $ denotes the boundary of the set $C$ and $|\xn - C| = \inf_{\rs \in C} |\xn - \rs| $, for a sequence $b_n \to 0$ so that   $\delta_n^{-1}(a_n + h_n) \to 0$ for any bandwidth $h_n$ that satisfies the assumptions. Then,
\begin{equation}
 \bar{y}_l(\V,\bs_0) = \frac{t_n^{(l)}(\V,\bs_0)}{t_n^{(0)}(\V,\bs_0)} =   \frac{t_n^{(l)}(\V,\bs_0)/t^{(0)}(\V,\bs_0)}{t_n^{(0)}(\V,\bs_0)/t^{(0)}(\V,\bs_0)} \label{ylbar}
\end{equation}
We consider the numerator and enumerator of \eqref{ylbar} separately. 
By Lemma~\ref{t_uniform} 
\begin{gather*}
\sup_{\V \times \bs_0 \in A_n} \left|\frac{t_n^{(0)}(\V,\bs_0)}{t^{(0)}(\V,\bs_0)} - 1\right| 
\leq \frac{\sup_{A}|t_n^{(0)}(\V,\bs_0) - t^{(0)}(\V,\bs_0)|}{\inf_{A_n} t^{(0)}(\V,\bs_0)} = O_P(\delta_n^{-1}(a_n + h_n))
\end{gather*}
\begin{gather*}
\sup_{\V \times \bs_0 \in A_n} \left|\frac{t_n^{(l)}(\V,\bs_0)}{t^{(0)}(\V,\bs_0)} - \mu_l(\V,\bs_0)\right| 
\leq \frac{\sup_{A}|t_n^{(l)}(\V,\bs_0) - t^{(l)}(\V,\bs_0)|}{\inf_{A_n} t^{(0)}(\V,\bs_0)} = O_P(\delta_n^{-1}(a_n + h_n)),
\end{gather*}
and therefore by $A_n \uparrow A = \spc(p,q) \times \text{supp}(f_\X)$,
\begin{equation*}
    \lim_{n \to \infty} \sup_{\V \times \bs_0 \in A_n}\left|\frac{t_n^{(l)}(\V,\bs_0)}{t^{(0)}(\V,\bs_0)} - \mu_l(\V,\bs_0)\right| = \lim_{n \to \infty} \sup_{\V \times \bs_0 \in A}\left|\frac{t_n^{(l)}(\V,\bs_0)}{t^{(0)}(\V,\bs_0)} - \mu_l(\V,\bs_0)\right| 
\end{equation*}
Substituting in \eqref{ylbar}, we obtain
\begin{equation*}
\bar{y}_l(\V,\bs_0) = \frac{t_n^{(l)}(\V,\bs_0)/t^{(0)}(\V,\bs_0)}{t_n^{(0)}(\V,\bs_0)/t^{(0)}(\V,\bs_0)} = \frac{\mu_l + O_P(\delta_n^{-1}(a_n + h_n))}{1 + O_P(\delta_n^{-1}(a_n + h_n))} = \mu_l + O_P(\delta_n^{-1}(a_n + h_n)).
\end{equation*}
For $l = 2$,  $\tilde{Y}^2_i = g(\widetilde{\B}^T\X_i)^2 + 2g(\widetilde{\B}^T\X_i)\tilde{\epsilon}_i + \tilde{\epsilon}_i^2$, and \eqref{Ltilde_uniform} follows from~\eqref{e_LtildeVs0}.
\end{proof}

\begin{lemma}\label{mu_lemma}
Under  (E.1), (E.2), (E.4), there exists $0 < C_5 < \infty$ such that
\begin{align}\label{mu_inequality}
    \left|\mu_l(\V,\bs_0) - \mu_l(\V_j,\bs_0)\right| \leq C_5 \|\Pbf_\V -\Pbf_{\V_j}\|
\end{align}
for all interior points $\bs_0 \in \text{supp}(f_\X)$
\end{lemma}
\begin{proof}
From  the representation $\tilde{t}^{(l)}(\Pbf_\V,\bs_0)$  in \eqref{Grassman} instead of $t^{(l)}(\V,\bs_0)$, we consider $\mu_l(\V,\bs_0) = \mu_l(\Pbf_\V,\bs_0)$ as a function on the Grassmann manifold since $\Pbf_\V \in Gr(p,q)$. 
Then,
\begin{align}\label{LipschitG}
    \left|\mu_l(\Pbf_\V,\bs_0) - \mu_l(\Pbf_{\V_j},\bs_0)\right| &= \left|\frac{\tilde{t}^{(l)}(\Pbf_\V,\bs_0)}{\tilde{t}^{(0)}(\Pbf_\V,\bs_0)} - \frac{\tilde{t}^{(l)}(\Pbf_{\V_j},\bs_0)}{\tilde{t}^{(0)}(\Pbf_{\V_j},\bs_0)}\right| \notag \\ 
   &\leq \frac{\sup |\tilde{t}^{(0)}(\Pbf_\V,\bs_0)|}{(\inf \tilde{t}^{(0)}(\Pbf_\V,\bs_0))^2}\left| \tilde{t}^{(l)}(\Pbf_\V,\bs_0)-\tilde{t}^{(l)}(\Pbf_{\V_j},\bs_0)\right|\notag \\
    &\quad +\frac{\sup \tilde{t}^{(l)}(\Pbf_\V,\bs_0)}{(\inf \tilde{t}^{(0)}(\Pbf_\V,\bs_0))^2}\left| \tilde{t}^{(0)}(\Pbf_\V,\bs_0)-\tilde{t}^{(0)}(\Pbf_{\V_j},\bs_0)\right| 
\end{align}
with $\sup_{\Pbf_\V \in Gr(p,q)} \tilde{t}^{(0)}(\Pbf_\V,\bs_0) \in (0,\infty)$ and  $\inf_{\Pbf_\V \in Gr(p,q)} \tilde{t}^{(0)}(\Pbf_\V,\bs_0) \in (0,\infty)$ since $\tilde{t}^{(l)}$ is continuous, $\Sigmaxbf >0$ and $\bs_0 \in \text{supp}(f_\X)$ an interior point.

By (E.2) and (E.4), $\tilde{g}(\xn) =g(\widetilde{\B}^T \xn)f_\X(\xn)$ is twice continuous differentiable  and therefore Lipschitz continuous on compact sets.  We denote its Lipschitz constant by $L < \infty$. Therefore,
\begin{gather} 
    \left| \tilde{t}^{(l)}(\Pbf_{\V},\bs_0)-\tilde{t}^{(l)}(\Pbf_{\V_j},\bs_0)\right| \leq \int_{\text{supp}(f_\X)} \left|\tilde{g}(\bs_0 + \Pbf_{\V} \rs)-\tilde{g}(\bs_0 + \Pbf_{\V_j} \rs)\right|d \rs  \notag \\ \leq
    L \int_{\text{supp}(f_\X)}  \|(\Pbf_{\V} -\Pbf_{\V_j}) \rs\|d\rs \leq
    L\left(\int_{\text{supp}(f_\X)}  \| \rs \|dr\right) \|\Pbf_\V -\Pbf_{\V_j}\| \label{t_inequality}
\end{gather}
where the last inequality is due to the sub-multiplicativity of the Frobenius norm and the integral being finite by (E.2). Plugging \eqref{t_inequality} in \eqref{LipschitG} and collecting all constants into $C_5$ yields \eqref{mu_inequality}. 
\end{proof}

\bigskip
\begin{proof}[Proof of Theorem~\ref{uniform_convergence_ecve}]
By \eqref{e_LN} and \eqref{e_objective},
\begin{align}
\left|L^*_n(\V,f) - L_\Fa^*(\V,f)\right| &\leq \left|\frac{1}{n} \sum_i \left(\tilde{L}_{n,\Fa}(\V,\X_i,f) -\tilde{L}_\Fa(\V,\X_i,f)\right)\right| \notag\\ &\qquad + \left|\frac{1}{n} \sum_i \left(\tilde{L}_\Fa(\V,\X_i,f) - \E(\tilde{L}_\Fa(\V,\X,f))\right) \right|   \label{Ln-L}
\end{align}
By Theorem~\ref{thm_Ltilde_uniform}, 
\begin{equation}
\left|\frac{1}{n} \sum_i \tilde{L}_{n,\Fa}(\V,\X_i,f) -\tilde{L}_\Fa(\V,\X_i,f)\right| \leq \sup_{\V \times \bs_0 \in A}\left|\tilde{L}_{n,\Fa}(\V,\bs_0,f) - \tilde{L}_\Fa(\V,\bs_0,f)\right| = o_P(1)
\end{equation}
The second term in \eqref{Ln-L} converges to 0 almost surely for all $\V \in \spc(p,q)$ by the strong law of large numbers. In order  to show uniform convergence the same technique as in the proof of Theorem~\ref{thm_variance} is used. Let $B_j = \{\V \in \spc(p,q): \|\V\V^T - \V_j\V_j^T\| \leq \tilde{a}_n\}$ be a cover of $\spc(p,q)\subset \bigcup_{j=1}^{N} B_j$ with $N \leq C\, \tilde{a}_n^{-d} = C\,(n/\log(n))^{d/2} \leq C\, n^{d/2}$, where $d = \dim(\spc(p,q))$ is defined in the proof of Theorem~\ref{thm_variance}. By Lemma~\ref{mu_lemma}, 
\begin{align}\label{inequality3}
    \left|\mu_l(\V,\X_i) - \mu_l(\V_j,\X_i)\right| \leq C_5 \|\Pbf_\V  - \Pbf_{\V_j}\|
\end{align}
Let $G_n(\V,f) = \sum_i\tilde{L}_\Fa(\V,\X_i,f)/n$ with $\E(G_n(V)) = L^*_\Fa(\V,f)$. Using \eqref{inequality3} and following the same steps as in the proof of Lemma~\ref{thm_variance} we obtain
\begin{align}\label{G_n_ineq}
    \left|G_n(\V,f)-L^*_\Fa(\V,f)\right| &\leq \left|G_n(\V,f) - G_n(\V_j,f)\right|\notag\\
    &\quad+ \left|G_n(\V_j,f) -L^*_\Fa(\V_j,f)\right| + \left|L^*_\Fa(\V,f)-L^*_\Fa(\V_j,f)\right| \notag \\ 
    &\quad \leq 2C_6\tilde{a}_n + \left|G_n(\V_j,f) -L^*_\Fa(\V_j,f)\right| 
\end{align}
for $\V \in B_j$ and some $C_6 > C_5$. Inequality \eqref{G_n_ineq} leads to
\begin{align} 
    \Pb\left(\sup_{\V \in \spc(p,q)}|G_n(\V,f) - L^*_\Fa(\V,f)| > 3C_6\tilde{a}_n\right) &\leq C\,N\, \Pb(\sup_{\V \in B_j}|G_n(\V,f) - L^*_\Fa(\V,f)| > 3C_6\tilde{a}_n) \notag \\ 
    &\leq C\, n^{d/2} \Pb(|G_n(\V_j,f) -L^*_\Fa(\V_j,f)|> C_6\tilde{a}_n) \notag \\
    &\leq C\, n^{d/2} n^{-\gamma(C_6)} \to 0 \label{inequality5}
\end{align}
where the last inequality in \eqref{inequality5} is due to \eqref{Bernstein} with $Z_i = \tilde{L}_\Fa(\V_j,\X_i,f)$, which is bounded since  $\tilde{L}_\Fa(\cdot,\cdot,f)$ is continuous on the compact set $A$, and $\gamma(C_6)$ a monotone increasing function of $C_6$ that can be made arbitrarily large by choosing $C_6$ accordingly. Therefore, $\sup_{\V \in \spc(p,q)}\left|L^*_n(\V,f) - L_\Fa^*(\V,f)\right| \leq o_P(1) + O_P(\tilde{a}_n)$ which implies  Theorem~\ref{uniform_convergence_ecve}.
\end{proof}
\begin{proof}[Proof of Theorem~\ref{thm_consistency_mean_subspace}]
We apply  \cite[Thm 4.1.1]{Takeshi} to obtain consistency of the conditional variance estimator. This theorem requires three conditions that guarantee the convergence of the minimizer of a sequence of random functions $L^*_n(\Pbf_\V,f_t)$ to the minimizer of the limiting function $L^*(\Pbf_\V,f_t)$; i.e., $\Pbf_{\spn\{\widehat{\B}^t_{{k_t}}\}^\perp} = \argmin L^*_n(\Pbf_\V,f) \to \Pbf_{\spn\{\B\}^\perp} = \argmin L^*(\Pbf_\V,f_t)$ in probability.
To apply the theorem three conditions have to be met:
(1) The parameter space 
is compact; 
(2) $L^*_n(\Pbf_\V,f_t)$ is continuous in $\Pbf_\V$ and a measurable function of the data $(Y_i,\X_i^T)_{i=1,...,n}$, and 
    (3) $L^*_n(\Pbf_\V,f_t)$ converges uniformly to $L^*(\Pbf_\V,f_t)$ and $L^*(\Pbf_\V,f_t)$ attains a unique global minimum at $\msf^\perp$. 

Since $L^*_n(\V,f_t)$ depends on $\V$ only through $\Pbf_\V = \V\V^T$,  $L^*_n(\V,f_t)$ can be considered as functions on the Grassmann manifold, which is compact, and the same holds true for $L^*(\V,f_t)$ by \eqref{Grassman}. 
Further,  $L^*_n(\V,f_t)$ is by definition a measurable function of the data and continuous in $\V$ if  a continuous kernel, such as the  Gaussian, is used. Theorem~\ref{uniform_convergence_ecve} obtains  the uniform convergence and Theorem~\ref{CVE_targets_meansubspace_thm} that the minimizer is  unique when $L(\V)$ is minimized  over the Grassmann manifold $G(p,q)$, since $\msf = \spn\{\widetilde{\B}\}$ is uniquely identifiable and so is   $\spn\{\widetilde{\B}\}^\perp$ (i.e. $\|\Pbf_{\spn\{\widehat{\B}^t_{k_t}\}} - \Pbf_{\spn\{\widetilde{\B}\}}\|=\|\widehat{\B}^t_{k_t}(\widehat{\B}^t_{k_t})^T - \widetilde{\B}\widetilde{\B}^T\| = \| (\I_p- \widetilde{\B}\widetilde{\B}^T)- (\I_p-\widehat{\B}^t_{k_t}(\widehat{\B}^t_{k_t})^T)\| = \| \Pbf_{\spn\{\widetilde{\B}\}^\perp}-\Pbf_{\spn\{\widehat{\B}^t_{k_t}\}^\perp}\|$). Thus, all three conditions are met and the result is obtained.
\end{proof}

\begin{proof}[Proof of Theorem~\ref{uniform_convergence_eobjective}]
Let $(t_j)_{j=1,\ldots,{m_n}}$ be an i.i.d. sample from $F_T$ and write
\begin{align}\label{L_nF_decompostion}
    |L_{n,\Fa}(\V) -L_{\Fa}(\V)| &= \left| \frac{1}{{m_n}} \sum_{j=1}^{{m_n}} \left(L_n^*(\V,f_{t_j}) - L^*(\V,f_{t_j})\right) \right| \notag \\
    &\quad +  \left| \frac{1}{{m_n}} \sum_{j=1}^{{m_n}} \left(L^*(\V,f_{t_j}) - \E_{t \sim F_T}(L^*(\V,f_{t})\right) \right|
\end{align}
Then, $\sup_{\V \in \spc(p,q)} \left| L_n^*(\V,f_{t}) - L^*(\V,f_{t}) \right| \leq 8M^2$, by the assumption $\sup_{t \in \Omega_T} |f_t(Y)| < M < \infty$, and the triangle inequality. That is,  $L_n^*(\V,f_{t})$ estimates a variance of a bounded response $ f_t(Y) \in [-M, M]$ and is therefore bounded by the squared range $4M^2$ of $f_t(Y)$. The same holds true for $L^*(\V,f_{t})$. Further, $8M^2$ is an integrable dominant function so that Fatou's Lemma applies.

Consider the first term on the right hand side of \eqref{L_nF_decompostion} and let $\delta > 0$.  By Markov's and triangle inequalities and Fatou's Lemma,
\begin{gather*}
   \limsup_{n} \Pb\left(\sup_{\V \in \spc(p,q)} \left| \frac{1}{{m_n}} \sum_{j=1}^{{m_n}} L_n^*(\V,f_{t_j}) - L^*(\V,f_{t_j}) \right| > \delta \right) \\
   \leq \frac{1}{\delta} \limsup_{n} \E_{F_T}\left(\E(\sup_{\V \in \spc(p,q)} \left| \frac{1}{{m_n}} \sum_{j=1}^{{m_n}} L_n^*(\V,f_{t_j}) - L^*(\V,f_{t_j}) \right|  \right) \; \mbox{:Markov inequality}\\
   \leq  \frac{1}{\delta}\limsup_{n}\E_{F_T}\left(  \frac{1}{{m_n}} \sum_{j=1}^{{m_n}} \E(\sup_{\V \in \spc(p,q)}|L_n^*(\V,f_{t_j}) - L^*(\V,f_{t_j})|\right) \\
   =  \frac{1}{\delta}\limsup_{n}\E_{F_T}\left(   \E(\sup_{\V \in \spc(p,q)}|L_n^*(\V,f_{t_j}) - L^*(\V,f_{t_j})|)\right) \\\\
   \leq  \frac{1}{\delta}\E_{ F_T}\left(   \E(\limsup_{n} \sup_{\V \in \spc(p,q)}|L_n^*(\V,f_{t_j}) - L^*(\V,f_{t_j})| \right) = \frac{1}{\delta}\E_{t\sim F_T}\left(   \E(0)\right) = 0
\end{gather*}
since by Theorem~\ref{uniform_convergence_ecve} it holds $\limsup_{n} \sup_{\V \in \spc(p,q)}|L_n^*(\V,f_{t_j}) - L^*(\V,f_{t_j})| = 0$.

\vspace{3mm}
For the second term on the right hand side of \eqref{L_nF_decompostion} we apply Theorem 2 of \cite{jennrich1969} in \cite[p. 40]{mickey1963test}: 
\begin{thm}\label{Jennrich}
Let $t_j$ be an i.i.d. sample and $L^*(\V,f_{t}): \Theta \times \Omega_T \to \real$ where $\Theta$ is a compact subset of an euclidean space. $L^*(\V,f_{t})$ is continuous in $\V$ and measurable in $t$ by Theorem~\ref{CVE_targets_meansubspace_thm}. If  $L^*(\V,f_{t_j})\leq h(t_j)$, where $h(t_j)$ is integrable with respect to $F_T$, then 
\begin{align*}
    \frac{1}{m_n}\sum_{j=1}^{m_n} L^*(\V,f_{t_j}) \longrightarrow \E_{F_T}\left(L^*(\V,f_{t}) \right) \quad \text{uniformly over $\V \in \Theta$ almost surely as $n \to \infty$ }
\end{align*}
\end{thm}
Here $\V \in \spc(p,q) = \Theta \subseteq \real^{pq}$, by $\sup_{t \in \Omega_T} |f_t(Y)| < M < \infty$ and an analogous argument as for the first term in \eqref{L_nF_decompostion}, $Z_j(\V) = L^*(\V,f_{t_j})< 4M^2$. Therefore, $\E(\sup_{\V\in \spc(p,q)} |Z_j(\V)|) < 4M^2 $, which is integrable. Further, since $t_j$ are an i.i.d. sample from $F_T$, $Z_j(\V)$ is a i.i.d. sequence of random variables, $Z_j(\V)$ is continuous in $\V$ by Theorem~\ref{CVE_targets_meansubspace_thm} and the parameter space $\spc(p,q)$ is compact. Then by Theorem \ref{Jennrich},
\begin{align*}
   \sup_{\V\in \spc(p,q)} \left| \frac{1}{{m_n}} \sum_{j=1}^{{m_n}} L^*(\V,f_{t_j}) - \E_{t \sim F_T}(L^*(\V,f_{t})) \right| \longrightarrow 0 \quad \quad \text{almost surely as $n \to \infty$}
\end{align*}
if $\lim_{n \to \infty} m_n = \infty$. Putting everything together it follows that $\sup_{\V\in \spc(p,q)} |L_{n,\Fa}(\V) -L_{\Fa}(\V)| \to 0$ in probability as $n \to \infty$.
\end{proof}
\begin{proof}[Proof of Theorem~\ref{ECVE_consistency}]
The proof is directly analogous to the proof of Theorem~\ref{thm_consistency_mean_subspace}. The uniform convergence of the target function $L_{n,\Fa}(\V)$ is obtained by Theorem~\ref{uniform_convergence_eobjective}. The  minimizer over $Gr(p,q)$ and its uniqueness derive from Theorem~\ref{ECVE_identifies_cs_thm}.
\end{proof}

\begin{proof}[Proof of Theorem~\ref{e_lemma-one}]
In this proof we supress the dependence on $f$ in the notation. The Gaussian kernel $K$ satisfies $\partial_z K(z) = - z K(z)$. From \eqref{weights} and \eqref{e_Ltilde} we have $\tilde{L}_{n,\Fa} = \bar{y}_2 - \bar{y}_1^2$ where $\bar{y}_l = \sum_i w_i \tilde{Y}_i^l$, $l=1,2$. We let $K_{j} = K(d_j(\V,\bs_0)/h_n)$,
suppress the dependence on $\V$ and $\bs_0$ and write $w_i = K_i/\sum_j K_j$. Then, $\nabla K_i = (-1/h_n^2)K_i d_i \nabla d_i$ and $\nabla w_i =  -\left(K_i d_i \nabla d_i (\sum_j K_j) - K_i \sum_j K_j d_j\nabla d_j\right)/(h_n \sum_j K_j)^2$. Next,
\begin{align}
\nabla \bar{y}_l &= -\frac{1}{h_n^2}\sum_i \tilde{Y}_i^l\frac{K_id_i\nabla d_i - K_i (\sum_jK_jd_j\nabla d_j)}{( \sum_jK_j)^2} 
= -\frac{1}{h_n^2}\sum_i \tilde{Y}_i^l w_i \left(d_i \nabla d_i - \sum_j w_j d_j \nabla d_j\right) \notag \\
&= -\frac{1}{h_n^2}\left(\sum_i \tilde{Y}_i^l w_i d_i \nabla d_i - \sum_j\tilde{Y}_j^l w_j \sum_i w_id_i \nabla d_i\right) 
=-\frac{1}{h_n^2}\sum_i (\tilde{Y}_i^l - \bar{y}_l) w_i d_i \nabla d_i \label{grad.yl}
\end{align}
Then, $\nabla \tilde{L}_n = \nabla \bar{y}_2 - 2\bar{y}_1 \nabla \bar{y}_1$, and inserting $\nabla \bar{y}_l$ from \eqref{grad.yl} yields $\nabla \tilde{L}_n = (-1/h_n^2)\sum_i (Y_i^2 -\bar{y}_2 - 2\bar{y}_1(Y_i - \bar{y}_1))w_i d_i \nabla d_i = (1/h_n^2)\left(\sum_i \left(\tilde{L}_n -(Y_i -\bar{y}_1)^2 \right) w_i d_i \nabla d_i\right)$, since $Y_i^2 -\bar{y}_2 - 2\bar{y}_1(Y_i - \bar{y}_1) = (Y_i -\bar{y}_1)^2 - \tilde{L}_n $.
\end{proof}

\end{document}